\documentclass{amsart}

\usepackage{amssymb}
\usepackage{enumerate}
\usepackage{subfig}
\usepackage{amsmath}
\usepackage{amsthm}
\usepackage{tikz} 
\usepackage{float}
\usepackage{mathtools}
\usepackage{graphicx}
\usepackage{lscape}
\usepackage{seqsplit}

\usepackage[normalem]{ulem}
\usepackage{pgf,tikz}\usepackage{mathrsfs}\usetikzlibrary{arrows}

\usetikzlibrary{matrix,fit}

\numberwithin{equation}{section}

\newtheorem{thm}{Theorem}[section]

\newtheorem{lem}[thm]{Lemma}
\newtheorem{rem}[thm]{Remark}

\newtheorem{obs}[thm]{Observation}

\theoremstyle{definition}
\newtheorem{defn}{Definition}[section]

\newcommand{\sn}{S_n}
\newcommand{\bn}{B_n}

\newcommand{\ul}[1]{\underline{#1}}
\newcommand{\Aut}{\mbox{Aut}}

\begin{document}

\author{Sa\'ul A. Blanco, Charles Buehrle, and Akshay Patidar}
\title[Cycles in the Burnt Pancake Graph]{Cycles in the burnt pancake graph}
\address{Department of Computer Science, Indiana University, Bloomington, IN 47408}
\email{sblancor@indiana.edu}
\address{ Department of Mathematics, Physics, and Computer Studies, Notre Dame of Maryland University, Baltimore, MD 21210}
\email{cbuehrle@ndm.edu}
\address{Department of Computer Science and Engineering, Indian Institute of Technology, Bombay, MH 400076}
\email{akshay@cse.iitb.ac.in}


\bibliographystyle{plain}


\begin{abstract}
The pancake graph $P_n$ is the Cayley graph of the symmetric group $\sn$ on $n$ elements generated by prefix reversals. $P_n$ has been shown to have properties that makes it a useful network scheme for parallel processors. For example, it is $(n-1)$-regular, vertex-transitive, and one can embed cycles in it of length $\ell$ with $6\leq\ell\leq n!$. The burnt pancake graph $BP_n$, which is the Cayley graph of the group of signed permutations $B_n$ using prefix reversals as generators, has similar properties. Indeed, $BP_n$ is $n$-regular and vertex-transitive. In this paper, we show that $BP_n$ has every cycle of length $\ell$ with $8\leq\ell\leq 2^n n!$. The proof given is a constructive one that utilizes the recursive structure of $BP_n$.

We also present a complete characterization of all the $8$-cycles in $BP_n$ for $n \geq 2$, which are the smallest cycles embeddable in $BP_n$, by presenting their canonical forms as products of the prefix reversal generators. 
\end{abstract}

\maketitle

\section{Introduction }\label{s:intro}

The \emph{pancake problem} was first introduced in Dweighter~\cite{Dweighter75}. The setting is as follows: There is a stack of pancakes of different diameters and one is to sort it utilizing a chef spatula. So the only operations that one can perform on the stack are prefix reversals: one picks up a substack from the top, flips it, and then places them back on top of the substack that was left. The pancake problem consists on finding the smallest number $f(n)$ of pancake flips or prefix reversals that are needed to sort any stack of $n$ pancakes, all of different diameter. The first nontrivial bound, $f(n)\leq (5n + 5)/3$, was given by Gates and Papadimitriou~\cite{GatesPapa} (when Gates was an undergraduate student), and this was later improved by Chitturi, Fahle, Meng, Morales, Shields, Sudborough, and Voit~\cite{Chitt} to $f(n)\leq (18/11)n+O(1)$. Finding a minimal length sequence of flips to sort a given stack is NP-hard as shown by Bulteau, Fertin, and Rusu~\cite{BulFerRusu}, and to our knowledge $f(n)$ is known for $1\leq n\leq 19$ (see Asai, Kounoike, Shinano, and Kaneko~\cite{Asai2006}, Cibulka~\cite{Cibulka}, Cohen and Blum~\cite{CohenBlum}, Heydari and Sudborough~\cite{HeySud}, and Kounoike, Kaneko, and Shinano~\cite{KKS}). Recently, a generalization of the pancake graph was used to construct a new general family of mixed graphs (see Dalf\'{o}~\cite{Dalfo}).

The burnt pancake problem introduced by Gates and Papadimitriou~\cite{GatesPapa} is concerned with finding the minimum number $g(n)$ of pancake flips, prefix reversals, needed to sort a stack of burnt on one side pancakes where all the burnt sides are down. Unlike the pancake problem, efficient algorithms exist to find the minimum number of flips needed to sort a burnt/signed stack (see Bergeron~\cite{B05} and Hannenhali and Pevzner~\cite{HannenPev}). However, exact values for the burnt pancake number are only known for $n\leq17$ (see Cibulka~\cite{Cibulka}). Cohen and Blum~\cite{CohenBlum} proved that $3n/2 \leq g(n) \leq 2n - 2$, for $n \geq 10$. In particular, the diameter of the burnt pancake graph is linear. 

Both the pancake problem and the burnt pancake problem give rise to graphs that have several nice properties. We use $P_n$ and $BP_n$ to denote the pancake graph and the burnt pancake graph of dimension $n$, respectively. Both $P_n$ and $BP_n$ are the Cayley graphs of the symmetric and hyperoctahedral group generated by prefix reversals, respectively. Since these graphs are Cayley graphs, they are vertex-transitive and both have a low diameter in comparison to the number of vertices of the graph. These properties have made the pancake graph an intriguing model for an interconnection network. Another desirable property that the pancake graph has is containing cycles of length $\ell$ with $6\leq\ell\leq\ n!$. Containing such cycles facilitates local connections within the network. 

Several models have been offered for interconnection schemes for parallel computers. Using Cayley graphs as interconnection schemes was first proposed in Akers and Krishnamurthy~\cite{AkersKrish}, and it offers several advantages due to their algebraic structure. For example, as mentioned earlier, every Cayley graph is vertex-transitive, which intuitively means that every vertex and its neighborhood looks the same as any other vertex. Hence, computation can be shared evenly among all the computers in the network. Furthermore, routing (communicating between the different computers in an interconnection network) is easy in Cayley graphs coming from certain permutation groups (see Schibell and Stafford~\cite{SS92}). In order to implement parallel algorithms designed for other architectures, it is desirable to be able to embed many types of subgraphs. At the very least the ability to embed cycles and paths of a given length is wanted. These cycles and paths, called \emph{rings} and \emph{lines} in the interconnection networks literature, are used for local communication and load balancing (see the discussion in Kanevsky and Feng~\cite{KF95} and references therein). Our proof is a recursive construction of all the possible cycle lengths.

The organization of the paper is as follows: The notation is explained in Section~\ref{s:notation}, where we also discuss and correct the original proof that the pancake graph has all cycles of length $\ell$ with $6\leq\ell\leq n!$. Later in Section~\ref{s:pancycle}, we provide a constructive, recursive proof of the existence of cycles in the burnt pancake graph from its girth which happens to be 8 (see Compeau~\cite{Compeau2011}) to its number of vertices, $2^n n!$. In Section~\ref{s:classification}, we provide a full characterization of the $8$-cycles inside the burnt pancake graph, inspired by analogous results for the pancake graph (see Konstantinova and Medvedev~\cite{KM10, KM11, KonMed}).  

\subsection{Main results} The main contributions of the paper are:
\begin{enumerate}
    \item We prove that $BP_n$, with $n\geq2$, has all cycles of length $\ell$ with $8\leq\ell\leq 2^nn!$. The details are discussed in Section~\ref{s:pancycle}. This result is surprising since burnt pancake graphs are fairly sparse. Furthermore, it shows that the burnt pancake graphs provide a viable interconnection scheme for parallel processes, since one would be able to use cycles or paths of a given length, as needed, for local communication. 
    \item We provide a characterization of all the $8$-cycles contained in any $BP_n$, with $n\geq2$. The details are given in Section~\ref{s:classification}. These would be the smallest local communication ring possible in such an interconnection scheme.
\end{enumerate}

To prove both these results, we use the recursive nature of the burnt pancake graph.

\section{Notation and preliminaries}\label{s:notation}

Let $\sn$ denote the group of permutations of the set $[n]:=\{1,2,\ldots, n\}$. The pancake problem has a natural interpretation in terms of $\sn$. If we denote a stack of pancakes using permutations, say for example one has the stack 4312 (the largest pancake is on top, then the second largest, then the smallest, and then the second smallest), one can turn that stack into 3412, 1342, and 2134 using \emph{prefix reversals} (permutations of the form $j\;(j-1)\cdots 1\;(j+1)\cdots n$ in one-line notation, with $1\leq j\leq n$). If we regard $\sn$ as being generated by prefix reversals (also known as \emph{pancake flips}), one can realize the \emph{pancake graph} as a Cayley graph of $\sn$ with generating set the set of prefix reversals. We denote this graph by $P_n$.

\subsection{Burnt pancake problem} The \emph{burnt pancake problem} introduced by Gates and Papadimitriou~\cite{GatesPapa} is concerned with finding the minimum number of pancake flips, or prefix reversals, needed to sort a stack of burnt on one side pancakes, where all the burnt sides are down. This is equivalent to finding the diameter of the Cayley graph for the hyperoctahedral group generated by prefix reversals. Specifically, the hyperoctahedral group $B_n$ is the group of signed permutations of $[\pm n]=\{-n,-(n-1),\ldots,-1,1,2,\ldots,n\}$ where $w \in \bn$ if and only if $w(-i)=-w(i)$ for all $i\in[\pm n]$. For ease of notation, it is usual to write $\ul{i}$ instead of $-i$. For $w\in\bn$, we utilize its \emph{window notation}, that is, we write $w=[w(1)\,w(2)\,\cdots\,w(n)]$. We denote a \emph{pancake flip} or \emph{prefix reversal} by 
$r_i$, for $1 \leq i \leq n$, where
\[r_i = [\ul{i}\,\ul{i-1}\,\ul{i-2}\cdots\,\ul{3}\,\ul{2}\,\ul{1}\,(i+1)\,(i+2)\,\cdots\,n].\]

A \emph{burnt pancake graph} $BP_n $ is defined as follows. $BP_n:= (\bn, E_n)$, where 
\begin{align*} 
    E_n = \Big\{ \big([w(1)\,\cdots\,w(i)\,w(i+1)\,\cdots\,w(n)],
    \,\,[\ul{w(i)}\,\cdots\,\ul{w(1)}\,w(i+1)\,\cdots\,w(n)]\big) \\
    \hfill:\,w \in \bn \text{ and } 1 \leq i \leq n \Big\}.
\end{align*}
That is, the vertices are all the signed permutations, and the edge set is the set of all pairs of permutations that are a pancake flip away from each other. One labels the edge 
$(w,wr_i)$ by $r_i$.

It is worth noting that $|\bn|=2^n n!$ and $|E_n|=n 2^{n-1} n!$, for $n \geq 1$.

\subsection{Recursive structure of $BP_n$}  Throughout the paper, we shall make constant use of the recursive structure of $BP_n$. We use the notation $BP_{n-1}(q)$ to denote the copy of $BP_{n-1}$ obtained by restricting $BP_n$ to its subgraph induced by all of those signed permutations whose last character is $q$ in the window notation, where $q\in[\pm n]$. Figures~\ref{fig:burnt2} and~\ref{fig:burnt3} show $BP_2$ and $BP_3$, and the recursive nature of these graphs is also showcased.

Furthermore, for $k < n$, we use $P_{k}(p)$ ($BP_{k}(p)$, respectively) to denote the subgraph of $P_{n-1}(n)$ ($BP_{n-1}(n)$, respectively) whose vertices are the set of all permutations $\pi \in S_{n}$ with $$\pi = \pi(1)\; \pi(2)\; \cdots \;\pi({k-1})\; p \;(k+1)\; (k+2)\; \cdots\; n,$$ where $p\in[k]$ and $\pi(i) \in [k] \setminus \{p\}$ or, respectively, $\pi \in B_{n}$ with $$\pi = [\pi(1)\; \pi(2)\; \cdots \;\pi({k-1})\; p\; (k+1)\; (k+2)\; \cdots \;n],$$ where $p \in [\pm k]$ and $\pi(i) \in [\pm k] \setminus \{p\}$. The edges of $P_{k}(p)$ being $\{(\pi,\pi r_i) : \pi \in P_{k}(p), \text{ for } 2 \leq i \leq k-1\}$ ($\{(\pi,\pi r_i) : \pi \in BP_{k}(p), \text{ for } i \in [k-1]\}$, respectively).

\begin{figure}
    \centering
    \begin{tikzpicture}[scale=0.05,line cap=round,line join=round,>=triangle 45,x=1.0cm,y=1.0cm]
	
	\begin{scriptsize}
	\node (e) at (-13,-33) {12};
	\node (0) at (13,-33) {\underline{1}2};
	\node (10) at (33,-14) {\underline{2}1};
	\node (010) at (33,14) {21};
	\node (1010) at (13,33) {$\underline{1}\,\underline{2}$};
	\node (101) at (-13,33) {$1\underline{2}$};
	\node (01) at (-33,14) {2\underline{1}};
	\node (1) at  (-33,-14) {\underline{2}\,\underline{1}};
    \end{scriptsize}
	\draw[line width=2] (e)--(0) (10)--(010) (101)--(1010) (1)--(01);
	\draw[line width=2,dashed] (e)--(1) (0)--(10) (01)--(101) (010)--(1010);

    \end{tikzpicture}
    \caption{$BP_2$ is an $8$-cycle.}
    \label{fig:burnt2}
\end{figure}
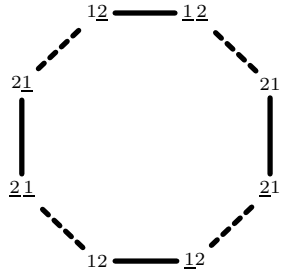

\begin{figure}
    \centering
   	\begin{tikzpicture}[scale=0.15,line cap=round,line join=round,>=triangle 45,x=1.0cm,y=1.0cm]
	
	\begin{scriptsize}
	\node (e) at (-13,-33) {123};
	\node (0) at (13,-33) {\underline{1}23};
	\node (10) at (33,-14) {\underline{2}13};
	\node (010) at (33,14) {213};
	\node (1010) at (13,33) {$\underline{1}\,\underline{2}3$};
	\node (101) at (-13,33) {$1\underline{2}3$};
	\node (01) at (-33,14) {2\underline{1}3};
	\node (1) at  (-33,-14) {\underline{2}\,\underline{1}3};

	\node (202) at  (-4,-9) {12\underline{3}};
	\node (0202) at  (4,-9) {\underline{1}2\underline{3}};
	\node (10202) at  (9,-3.5) {\underline{2}1\underline{3}};
	\node (010202) at  (9,4) {21\underline{3}};
	\node (1010202) at  (4,9) {$\underline{1}\;\underline{2}\,\underline{3}$};
	\node (101202) at  (-4,9) {$1\underline{2}\,\underline{3}$};
	\node (01202) at  (-9,4) {$2\underline{1}\;\underline{3}$};
	\node (1202) at  (-9,-3.5) {$\underline{2}\,\underline{1}\,\underline{3}$};

	\node (21) at  (-12.5,-20) {\underline{3}12};
	\node (121) at  (-6,-15) {\underline{1}32};
	\node (0121) at  (6,-15) {132};
	\node (10121) at  (12.5,-20) {\underline{3}\,\underline{1}\,2};
	\node (010121) at  (11,-25) {3\underline{1}2};
	\node (01021) at  (7,-29.5) {1\underline{3}2};
	\node (1021) at  (-7,-29.5) {\underline{1}\,\underline{3}2};
	\node (021) at  (-11,-25) {312};

	\node (212) at (6,14.75) {1\underline{3}\,\underline{2}};
	\node (0212) at (-6,14.75) {\underline{1}\,\underline{3}\,\underline{2}};
	\node (10212) at (-12.5,20) {31\underline{2}};
	\node (010212) at (-12,25) {\underline{3}1\underline{2}};
	\node (1010212) at (-7,29.5) {\underline{1}3\underline{2}};
	\node (101212) at (7,29.5) {13\underline{2}};
	\node (01212) at (12,25) {\underline{3}\,\underline{1}\,\underline{2}};
	\node (1212) at (12.5,20) {3\underline{1}\,\underline{2}};

	\node (20) at  (29.5,-8) {\underline{3}\,\underline{2}1};
	\node (020) at  (29.5,7) {3\underline{2}1};
	\node (1020) at  (26,11.5) {2\underline{3}1};%
	\node (01020) at  (18.5,11.5) {\underline{2}\,\underline{3}1};
	\node (101020) at  (17,5.5) {321};
	\node (10120) at  (16.5,-7) {\underline{3}\,21};	
	\node (0120) at  (18.5,-11.5) {\underline{2}\,31};
	\node (120) at  (24.5,-11.5) {231};

	\node (2) at  (-29.5,-8) {\underline{3}\,\underline{2}\,\underline{1}};
	\node (02) at  (-29.5,7) {3\underline{2}\,\underline{1}};
	\node (102) at  (-26,11.5) {2\underline{3}\,\underline{1}};%
	\node (0102) at  (-18.5,11.5) {\underline{2}\,\underline{3}\,\underline{1}};
	\node (10102) at  (-17,5.5) {32\underline{1}};
	\node (1012) at  (-16.5,-7) {\underline{3}2\underline{1}}; 
	\node (012) at  (-18.5,-11.5) {\underline{2}3\underline{1}};
	\node (12) at  (-24.5,-11.5) {23\underline{1}};	
	\end{scriptsize}

	\draw[line width=2,dashed] (e)--(0) (10)--(010) (101)--(1010) (1)--(01) (202)--(0202) (10202)--(010202) 
	(101202)--(1010202) (1202)--(01202) (21)--(021) (121)--(0121) (10121)--(010121) (1021)--(01021) 
	(212)--(0212) (10212)--(010212) (101212)--(1010212) (1212)--(01212) (20)--(020) (1020)--(01020) 
	(10120)--(101020) (120)--(0120) (2)--(02) (102)--(0102) (1012)--(10102) (12)--(012);

	\draw[line width=2] (e)--(1) (0)--(10) (2)--(12) (01)--(101) (02)--(102) (20)--(120) 
	(21)--(121) (010)--(1010) (012)--(1012) (020)--(1020) (021)--(1021) (010212)--(1010212) 
	(0121)--(10121) (202)--(1202) (212)--(1212) (0202)--(10202) (010202)--(1010202) (01202)--(101202) 
	(0102)--(10102) (01021)--(010121) (01020)--(101020) (0120)--(10120) (01212)--(101212) (0212)--(10212);

	\draw[line width=1,dotted] (e)--(2) (0)--(20) (1)--(21) (10)--(10121) (02)--(202) (010)--(01212) 
	(1010)--(10120) (101)--(1012) (01)--(010212) (020)--(0202)  (010121)--(10202) (10102)--(101202)
	(101020)--(1010202) (120)--(0212) (121)--(01020) (12)--(212) (0121)--(0102) (010202)--(1212)
	(01021)--(012) (1021)--(0120) (021)--(1202) (101212)--(102) (1010212)--(1020) (01202)--(10212);
	
    \end{tikzpicture}
    \caption{$BP_3$ has 6 embedded copies of $BP_2$ obtained by looking at the subgraphs induced by fixing the last character in the window notation to be $\pm1,\pm2$, or $\pm3$. All of the different copies are connected to each other using edges labeled by 
    $r_3$, depicted with a dotted line. }
    \label{fig:burnt3}
\end{figure}
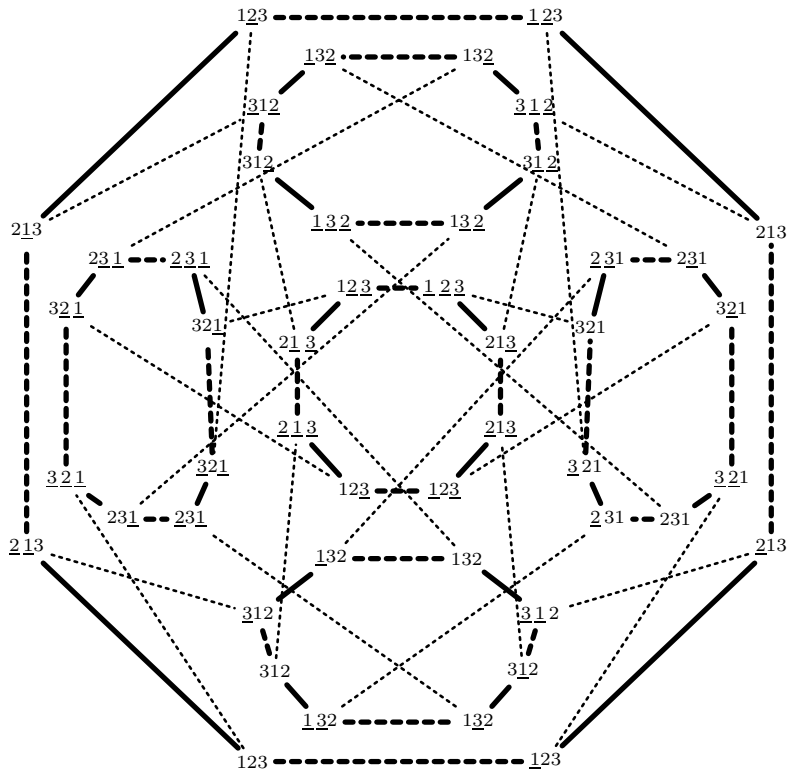

\subsection{The pancake graph $P_n$ is weakly pancyclic}

A graph $G=(V,E)$ is said to be \emph{pancyclic} if it has all cycles of length $\ell$, for $3\leq\ell\leq|V|$. Furthermore, the \emph{girth} of $G$ is the length of the shortest cycle contained in $G$, and its \emph{circumference} is the length of the longest cycle contained in $G$. The graph $G$ is called \emph{weakly pancyclic} if it has cycles of all lengths including and between its girth and its circumference. It is known that the girth of $P_n$ is six (see Kanevsky and Feng~\cite{KF95}) and the girth of the burnt pancake graph $BP_n$ is eight (see Compeau~\cite[Theorem 10]{Compeau2011}). Furthermore, $P_n$ is also weakly pancyclic (see Kanevsky and Feng~\cite{KF95} and Sheu, Tan, and Chu~\cite{STC06}. In fact, $P_n$ has a slightly stronger property than just being weakly pancyclic, since it is also Hamiltonian, that is, it $P_n$ has a Hamiltonian cycle. 

The main result of this paper is that $BP_n$ is Hamiltonian and weakly pancyclic: If $n\geq2$, $BP_n$ has a cycle of any length $\ell$, where $8\leq\ell\leq 2^nn!$. Our proof follows a similar technique to that in Kanevsky and Feng~\cite[Theorem 1]{KF95}, where the authors show that $P_n$, for $n\geq3$, has cycles of length from 6 (the girth of $P_n$) up to $n!-2$, and it is also Hamiltonian having a cycle of length $n!$. They provide a constructive, recursive proof, and our proof is in the same spirit. In regard to the proof in Kanevsky and Feng~\cite{KF95}, some remarks are in order.
\begin{rem}

\begin{enumerate}
\item That $P_n$ has a cycle of length $n!-1$ was not proven in Kanevsky and Feng~\cite{KF95}, though their proof can easily be extended to include that case as well. The existence of such a cycle of length $n!-1$ was first established from the main result in Hung, Hsu, Liang, and Hsu~\cite{Hung2003}, where they proved a stronger result.
    
\item The main reason of why the proof in Kanevsky and Feng~\cite{KF95} does not cover the existence of a cycle of length $n!-1$ is that they did not provide a cycle of length $23$ contained in $S_4$ to be the basis of their induction. They do provide cycles from length 6 up to 22 and 24 contained in $P_4$. However, a computer search shows that there are 184 different 23-cycles in this graph. 
Such a cycle is \[r_2 r_3 r_2 r_3 r_2 r_4 r_2 r_3 r_4 r_3 r_2 r_4 r_2 r_3 r_2 r_4 r_3 r_2 r_4 r_2 r_4 r_2 r_4.\]

\item In their proof, there is also an easily rectifiable mistake in case (5), regarding the cycles of length $\ell=a \cdot (n-1)! + b$ for $1 \leq a \leq n-1$ and $0 \leq b \leq (n-1)!-1$. Specifically, when these authors are considering $a \leq n-3$ and $b \leq (n-2)!+6$, they build a cycle of length $\ell$ beginning with their base cycle with $k=a+1$ by attaching Hamiltonian cycles to the base cycle in all but the first two copies of $BP_{n-1}$ and the last two copies of $BP_{n-1}$. In the second copy, they attach a cycle of length $(n-1)!-(n-2)!$; and in the second to the last copy, they attach a cycle of length $(n-2)!+b-5$. However, a length of $(n-2)!+b-5$ does not guarantee that an edge labeled by $r_{n-1}$ would be present since \[ (n-2)!+b-5 \geq (n-2)!-5,\] which clearly is not greater than $(n-2)!$. A fix to this issue is to attach a cycle of length $(n-1)!-(n-2)!-5$ in the second copy and attach a cycle of length $(n-2)!$ in the second to the last copy. This adjustment will force both such cycles to have sufficient length to have the needed edge, since $n \geq 5$ implies $(n-1)!-(n-2)! \geq 18$ and $(n-1)!-(n-2)!-5 \geq (n-2)!$. That is, there is room to remove five edges from the cycle and it would still have the needed edge. Moreover, this adjustment will not affect the overall length $\ell$.
\end{enumerate}
\end{rem}

\section{The Burnt pancake graph is Hamiltonian and weakly pancyclic}\label{s:pancycle}

Here we present a more precise definition of vertex- and edge-transitive graphs, and present our main result.

\begin{defn}[Vertex- and edge-transitive]
Given a graph $G=(V,E)$, we say that $G$ is 
\begin{description}
    \item[vertex-transitive.] If for any pair $v_1,v_2\in V$, there exists some $\varphi\in \Aut(G)$ such that $\varphi(v_1)=v_2$. Moreover, if $V'\subseteq V$ and for every $v'_1,v'_2\in V'$, there exists $\varphi'\in\Aut(G)$ such that $\varphi'(v'_1)=v'_2$, we say that $V'$ is \emph{vertex-transitive}. 
    \item[edge-transitive.] If for any pair $e_1,e_2\in E$, there exists some $\psi\in \Aut(G)$ such that $\psi(e_1)=e_2$. Moreover, if $E'\subseteq E$ and for every $e'_1,e'_2\in E'$, there exists $\psi'\in\Aut(G)$ such that $\psi'(e'_1)=e'_2$, we say that $E'$ is \emph{edge-transitive}. 
\end{description} 
\end{defn}

Since both $P_n$ and $BP_n$ are Cayley graphs, they are both vertex-transitive. Furthermore, $P_n$ and $BP_n$ are both not edge-transitive. That $P_n$ are not edge-transitive is mentioned in many places, for example, see Kanevsky and Feng~\cite{KF95} or Lakshmivarahan, Jwo, and Dhall~\cite{LJD93}. Furthermore, if $BP_3$ were edge-transitive, then there must be an automorphism $\phi$ such that 
$\phi(r_1)=r_2$. Hence, since $\phi$ is one-to-one and onto, 
$\phi(r_3)=r_3$. Therefore, the pair of edges 
$r_1,r_3$ is mapped to 
$r_2,r_3$. Direct inspection of Figure~\ref{fig:burnt3} gives that 
$r_1,r_3$ form an $8$-cycle whereas 
$r_2,r_3$ form a 12-cycle, and therefore no such $\phi$ can exist. Since $BP_3$ is contained in $BP_n$ with $n\geq3$, then $BP_n$ cannot be edge-transitive. 

In spite of $P_n$ and $BP_n$ not being edge-transitive, certain sets of their edges are. Indeed, the set of edges labeled by the same generators is edge-transitive. We prove this edge-transitivity result in the lemma below. As it is customary, if $G$ is a graph, we denote its vertex and edge set by $V(G)$ and $E(G)$, respectively.

\begin{lem}\label{l:edgetrans}
 Let $E^k_n$ be the set of edges of $BP_n$ labeled by 
 $r_k$, with $1\leq k\leq n$. Then, for every pair $e_1,e_2\in E^k_n$, there exists $\varphi\in \Aut(BP_{n})$ such that $\varphi(e_1)=e_2$.
\end{lem}
\begin{proof}
Recall that the vertices of $BP_n$ are signed permutations. Let $e_1,e_2\in E^k_n$, and therefore the edges $e_1,e_2$ have the form 
$(\pi_1,\pi_1r_k)$ and $(\pi_2,\pi_2r_k)$, respectively. Let $\varphi:V(BP_n)\to V(BP_n)$ be given by $\varphi(x)=\pi_2\pi_1^{-1}x$. It is easy to see that $\varphi$ is an automorphism of $BP_n$ for $e=(\pi,\pi')\in E(BP_n)$ if and only if 
$\pi'=\pi r_m$ for some $1\leq m\leq n$, and if and only if $(\varphi(\pi),\varphi(\pi'))\in E(BP_n)$. Notice that $\varphi(\pi_1)=\pi_2$ and 
$\varphi(\pi_1r_k)=\pi_2r_k$, and so $\varphi(e_1)=e_2$.
\end{proof}

A starting point for constructing arbitrary length cycles is to define what are referred to as \emph{base cycles}, which include an edge in any number of the copies of $BP_{n-1}$ embedded in $BP_n$. To denote these cycles, we label the edges by the pancake generators of $B_n$, where we assume that $n\geq4$. The base cycles are described below. 
    \[
        C_k = \begin{cases}
            r_{k+1} \left(r_{n}r_{n-1}\right)^{2k+2}r_{n-2k-2}r_{n}\left(r_{k+1}r_{k}\right)^{k+1}r_{2k+2}\left(r_{k}r_{k+1}\right)^{k}r_{k}, \\
            \hfill \text{for }1 \leq k < \left\lceil n/2 \right\rceil; \\
            r_{k+1} \left(r_{n}r_{n-1}\right)^{2k+2}r_{n}r_{k+1}\left(r_{n-k-2}r_{n-k-1}\right)^{n-k-2}r_{n-k-2}r_{n} \\
            \qquad \left(r_{n-k-2}r_{n-k-1}\right)^{n-k-2}r_{n-k-2}, \hfill \text{for } \left\lceil n/2 \right\rceil \leq k < n-1; \\
            \left(r_{n}r_{n-1}\right)^{2n}, \hfill \text{for } k=n-1,n.
            \end{cases}
    \]
    
\begin{lem}\label{l:basecycles}
    For $k \in [n]$, the paths $C_{k}$ defined above are cycles in $BP_{n}$.
\end{lem}
\begin{proof}
    Without loss of generality, we can assume that $C_{k}$ begins at the identity $e=[1\,2\,\cdots\,n]$. Direct computation of the resulting signed permutation after each prefix reversal verifies that $C_{k}$ is a cycle for the three cases of the value of $k$ relative to $n$. In Tables~\ref{tab:basecycle1}, \ref{tab:basecycle2}, and \ref{tab:basecycle3} the resulting signed permutation is beside the reversal(s) which was applied to the previous line's signed permutation.
    
    In the paths $C_{k}$, there are several instances of reversals that alternate the labels $r_{i}$ and $r_{i-1}$, in that order, with $1 < i \leq n$. When such a pair of reversals is applied to a given signed permutation, say $w = [w(1)\,w(2)\,\cdots\,w(n)]\in\bn$, the resulting signed permutation is 
    \begin{align*}
        w r_{i} &= [\ul{w(i)}\,\cdots\,\ul{w(1)}\,w(i+1)\,\cdots\,w(n)]\\
        w r_{i}r_{i-1} &= [w(2)\,\cdots\,w(i)\,\ul{w(1)}\,w(i+1)\,\cdots\,w(n)].
    \end{align*}
    Thus, the first character is moved to the $i$th position, with opposite sign, and all of the other characters are in their original relative positions with their original signs.
    
    Alternatively, when the order of the reversals is switched, that is, the labels $r_{i-1}$ and $r_{i}$, in that order, with $1 < i \leq n$, are applied to a signed permutation, say $w = [w(1)\,w(2)\,\cdots\,w(n)]\in\bn$, the resulting signed permutation is
    \begin{align*}
        w r_{i-1} &= [\ul{w(i-1)}\,\cdots\,\ul{w(1)}\,w(i)\,\cdots\,w(n)]\\
        w r_{i-1}r_{i} &= [\ul{w(i)}\,w(1)\,\cdots\,w(i-1)\,w(i+1)\,\cdots\,w(n)].
    \end{align*}
    Thus, the $i$th character is moved to the first position, with opposite sign, and all of the other characters are in their original relative positions with their original signs.
    
    When $1 \leq k < \lceil n/2 \rceil$, one can see in Table~\ref{tab:basecycle1}, following the $r_{k+1}$ reversal, there is a sequence of $4k+4$ reversals that alternate the labels $r_{n}$ and  $r_{n-1}$, in that order, that moves the first $2k+2$ characters to the end of the window notation and changes their signs. Following that, the remaining $n-2k-2$ characters are reversed and the entire window is reversed. Then a sequence of $2k$ reversals labeled by alternating $r_{k+1}$ and $r_{k}$, in that order, keeps the first $k$ characters in the same positions but with opposite signs. The first $2k+2$ characters are reversed, simultaneously placing the last $n-k-1$ in sorted order with correct signs. Then a sequence of $2k$ reversals, labeled by alternating $r_{k}$ and $r_{k+1}$, in that order, moves the $k$ characters following character $(k+1)$ in front of that character, with signs reversed. Finally the first $k$ characters are reversed, returning back to the identity.
    
    When $\lceil n/2 \rceil \leq k < n$, one can see in Table~\ref{tab:basecycle2} that the path begins in the same manner as above, but since $2k+2 > n$ the sequence of $4k+4$ reversals, labeled by alternating $r_{n}$ and $r_{n-1}$, in that order, will move all of the characters, at least once, and $2k+2-n$ characters are moved twice. Then the entire window is reversed followed by a reversal of the first $k+1$ characters. Then a sequence of $2(n-k-2)$ reversals, labeled by alternating $r_{n-k-2}$ and $r_{n-k-1}$, in that order, moves the $n-k-2$ characters in the second position to the $(n-k-1)$th position to the beginning of the window, with opposite signs. The reversal $r_{n-k-2}$ then reverses those characters just moved by the prior sequence. The entire window is reversed with an $r_{n}$. Then a sequence of $2(n-k-2)$ reversals, labeled by alternating $r_{n-k-2}$ and $r_{n-k-1}$, in that order, moves the $n-k-2$ characters in the second position to the $(n-k-1)$th position to the beginning of the window, with opposite signs. Lastly, the first $n-k-2$ characters are reversed to result in the identity.
    
    Finally, the path $C_n$, as can be seen in Table~\ref{tab:basecycle3}, begins with the identity and, after applying $(r_{n}r_{n-1})^{n}$, the $n$ characters of the identity in window notation move to the end of the window, with opposite signs. Thus, the original ordering of the characters is kept but with opposite signs. Then applying $(r_{n}r_{n-1})^{n}$ repeats the same process, which results in the identity.
    
    Since all three types of $C_k$, for $k \in [n]$, begin and end at the identity, then all $C_k$, for $k \in [n]$, are cycles in $BP_{n}$.
\end{proof}
 
The following results are immediate from the construction of the base cycles.

\begin{lem}\label{l:baselength}
The length of the $C_k$ is $8k+11$, when 
$1 \leq k < \left\lceil n/2 \right\rceil$; 
$4n+2$, when $\left\lceil n/2 \right\rceil \leq k < n-1$; 
 and $4n$, when $k=n-1,n$. 
\end{lem}

\begin{lem}\label{l:copiesmet}
 
    When 
    $1 \leq k < \left\lceil n/2 \right\rceil$, $C_k$ goes through $2k+3$ different copies of $BP_{n-1}$, where there is a single edge in $2k+1$ copies, two edges in a single copy, and $4k+5$ edges in a single copy.
    
    When 
    $\left\lceil n/2 \right\rceil \leq k < n-1$, $C_k$ goes through $2k+4$ different copies of $BP_{n-1}$, where there is a single edge in $2k+2$ copies, $2n-2k-2$ edges in a single copy, and $2n-2k-2$ edges in a single copy.
    
    When 
    $k=n-1,n$, $C_k$ goes through $2n$ different copies of $BP_{n-1}$, where there is a single edge in each copy.
    
\end{lem}

\begin{table}
    \centering
    \begin{tabular}{l|l}
         \hline
        Edge label & Node \\ \hline
         & $[1\,2\,3\,\ldots\,(k-1)\,k\,(k+1)\,(k+2)\,\ldots\,(n-1)\,n]$ \\
        $r_{k+1}$ & $[\ul{k+1}\,\ul{k}\,\ul{k-1}\,\ldots\,\ul{3}\,\ul{2}\,\ul{1}\,(k+2)\,\ldots\,(n-1)\,n]$ \\
        $r_{n}$ & $[\ul{n}\,\ul{n-1}\,\ldots\,\ul{k+2}\,1\,2\,3\,\ldots\,(k-1)\,k\,(k+1)]$ \\
        $r_{n-1}$ & $[\ul{k}\,\ul{k-1}\,\ldots\,\ul{3}\,\ul{2}\,\ul{1}\,(k+2)\,\ldots\,(n-1)\,n\,(k+1)]$ \\
        $\vdots$&\qquad $\vdots$\\
        $ \left(r_{n-1}r_{n}\right)^{2k}$ & $[(2k+2)\,\ldots\,n\,(k+1)\,k\,\ldots\,1\,\ul{k+2}\,\ul{k+3}\,\ldots\,\ul{2k+1}]$  \\
        $r_{n}$ & $[(2k+1)\,\ldots\,(k+2)\,\ul{1}\,\ldots\,\ul{k}\,\ul{k+1}\,\ul{n}\,\ldots\,\ul{2k+2}]$ \\
        $r_{n-1}$ & $[(2k+3)\,(2k+4)\,\ldots\,n\,(k+1)\,k\,\ldots\,1\,\ul{k+2}\,\ul{k+3}\,\ldots\,\ul{2k+2}]$ \\
        $r_{n-2k-2}$ & $[\ul{n}\,\ldots\,\ul{2k+4}\,\ul{2k+3}\,(k+1)\,k\,\ldots\,1\,\ul{k+2}\,\ul{k+3}\,\ldots\,\ul{2k+2}]$ \\
        $r_{n}$ & $[(2k+2)\,\ldots\,(k+2)\,\ul{1}\,\ldots\,\ul{k}\,\ul{k+1}\,(2k+3)\,(2k+4)\,\ldots\,n]$ \\
        $r_{k+1}$ & $[\ul{k+2}\,\ul{k+3}\,\ldots\,\ul{2k+2}\,\ul{1}\,\ldots\,\ul{k}\,\ul{k+1}\,(2k+3)\,(2k+4)\,\ldots\,n]$ \\
        $r_{k}$ & $[(2k+1)\,(2k)\,\ldots\,(k+2)\,\ul{2k+2}\,\ul{1}\,\ldots\,\ul{k}\,\ul{k+1}\,(2k+3)\,\ldots\,n]$ \\
        $\vdots$&\qquad $\vdots$\\
        $ \left(r_{k}r_{k+1}\right)^{k-2}$ & \qquad$\vdots$  \\
        $r_{k+1}$ & $[(k+3)\,\ldots\,(2k+2)\,\ul{k+2}\,\ul{1}\,\ul{2}\ldots
        \,\ul{k}\,\ul{k+1}\,(2k+3)\,\ldots\,n]$ \\
        $r_{k}$ & $[\ul{2k+2}\,\ul{2k+1}\,\ldots\,\ul{k+2}\,\ul{1}\,\ldots\,\ul{k}\,\ul{k+1}\,(2k+3)\,\ldots\,n]$ \\
        $r_{2k+2}$ & $[(k+1)\,k\,\ldots\,1\,(k+2)\,(k+3)\,\ldots\,(2k+2)\,\ldots\,n]$ \\
        $r_{k}$ & $[\ul{2}\,\ldots\,\ul{k}\,\ul{k+1}\,1\,(k+2)\,(k+3)\,\ldots\,n]$ \\
        $r_{k+1}$ & $[\ul{1}\,(k+1)\,k\,\ldots\,2\,(k+2)\,(k+3)\,\ldots\,n]$ \\
       $ \vdots $&\qquad $\vdots$\\
        $\left(r_{k+1}r_{k}\right)^{k-2}$ &\qquad$\vdots$  \\
        $r_{k}$ & $[\ul{k+1}\,1\,2\,\,3\,\ldots\,(k-1)\,k\,(k+2)\,\ldots\,(n-1)\,n]$ \\
        $r_{k+1}$ & $[\ul{k}\,\ul{k-1}\,\ldots\,\ul{3}\,\ul{2}\,\ul{1}\,(k+1)\,(k+2)\,\ldots\,(n-1)\,n]$ \\
        $r_{k}$ & $[1\,2\,3\,\ldots\,(k-1)\,k\,(k+1)\,(k+2)\,\ldots\,(n-1)\,n]$
    \end{tabular}
    \caption{Base cycle $C_k$ when $1 \leq k < \left\lceil n/2 \right\rceil$. If $k=1$ or $2$, the $(r_{k}r_{k-1})^{k-2}$ portion on the cycle is empty. }
    \label{tab:basecycle1}
\end{table}

\begin{table}
    \centering
    \begin{tabular}{l|l}
         \hline
        Edge label & Node \\ \hline
         & $[1\,2\,3\,\ldots\,(k-1)\,k\,(k+1)\,(k+2)\,\ldots\,(n-1)\,n]$ \\
        
        $r_{k+1}$ & $[\ul{k+1}\,\ul{k}\,\ul{k-1}\,\ldots\,\ul{3}\,\ul{2}\,\ul{1}\,(k+2)\,\ldots\,(n-1)\,n]$ \\
        
        $r_{n}$ & $[\ul{n}\,\ul{n-1}\,\ldots\,\ul{k+2}\,1\,2\,3\,\ldots\,(k-1)\,k\,(k+1)]$ \\
        
        $r_{n-1}$ & $[\ul{k}\,\ul{k-1}\,\ldots\,\ul{3}\,\ul{2}\,\ul{1}\,(k+2)\,\ldots\,(n-1)\,n\,(k+1)]$ \\
        
        $\vdots$ &\qquad $\vdots$\\
        
        $\left(r_{n-1}r_{n}\right)^{2k}$ & $[(n-k)\;\ldots1\;\ul{k+2}\;\ldots
        \,\ul{n}\;\ul{k+1}\;\ul{k}\ldots
        $\\ & \hfill $\ldots
        \,\ul{n-k}\,\ul{n-k+1}]$ \\
        
        $r_{n}$ & $[(n-k+1)\,\ldots\,k\,(k+1)\,n\,\ldots
        \,(k+2)\,\ul{1}\,\ldots\,\ul{n-k}]$ \\
        
        $r_{n-1}$ & $[(n-k-1)\,\ldots\,1\,\ul{k+2}\,\ldots
        $ \\ & \hfill $\ldots
        \,\ul{n}\,\ul{k+1}\,\ul{k}\,\ldots\,\ul{n-k+1}\,\ul{n-k}]$ \\
        
        $r_{n}$ & $[(n-k)\,(n-k+1)\,\ldots$\\ & \hfill $\ldots\,(k+1)\,n\,\ldots
        \,(k+2)\,\ul{1}\,\ldots\,\ul{n-k-1}]$ \\
        
        $r_{k+1}$ & $[\ul{k+2}\,\ldots\,\ul{n}\,\ul{k+1}\,\ul{k}\,\ldots$\\ & \hfill $\ldots\,\ul{n-k+1}\,\ul{n-k}\,\ul{1}\,\ldots
        \,\ul{n-k-1}]$ \\
        
        $r_{n-k-2}$ & $[(n-1)\,\ldots\,(k+2)\,\ul{n}\,\ul{k+1}\,\ul{k}\,\ldots$\\ & \hfill $\ldots\,\ul{n-k+1}\,\ul{n-k}\,\ul{1}\,\ldots
        \,\ul{n-k-1}]$ \\
        
        $r_{n-k-1}$ & $[n\,\ul{k+2}\,\ldots\,\ul{n-1}\,\ul{k+1}\,\ul{k}\,\ldots$\\ & \hfill $\ldots\,\ul{n-k+1}\,\ul{n-k}\,\ul{1}\,\ldots
        \,\ul{n-k-1}]$ \\
        
        $\vdots$&\qquad$\vdots$\\
        
        $\left(r_{n-k-1} r_{n-k-2}\right)^{n-k-4}$ &  \\
        
        $r_{n-k-2}$ & $[\ul{n}\,\ldots\,\ul{k+2}\,\ul{k+1}\,\ul{k}\,\ldots
        $\\ & \hfill $\ldots
        \,\ul{n-k+1}\,\ul{n-k}\,\ul{1}\,\ldots
        \,\ul{n-k-1}]$ \\
        
        $r_{n-k-1}$ & $[(k+3)\,(k+4)\,\ldots\,n\,\ul{k+2}\,\ul{k+1}\,\ldots$\\ & \hfill $\ldots\,\ul{n-k+1}\,\ul{n-k}\,\ul{1}\,\ldots
        \,\ul{n-k-1}]$ \\
        
        $r_{n-k-2}$ & $[\ul{n}\,\ldots\,\ul{k+3}\,\ul{k+2}\,\ul{k+1}\,\ul{k}\,\ldots$\\ & \hfill $\ldots\,\ul{n-k+1}\,\ul{n-k}\,\ul{1}\,\ldots
        \,\ul{n-k-1}]$ \\
        
        $r_{n}$ & $[(n-k-1)\,(n-k-2)\,\ldots$\\ & \hfill $\ldots\,1\,(n-k)\,(n-k+1)\,\ldots\,
        \,k\,(k+1)\,\ldots\,n]$ \\
        
        $r_{n-k-2}$ & $[\ul{2}\,\ul{3}\,\ldots\,\ul{n-k-1}\,1\,(n-k)\,(n-k+1)\,\ldots$\\ & \hfill $\ldots\,k\,(k+1)\,\ldots\,n]$ \\
        
        $r_{n-k-1}$ & $[\ul{1}\,(n-k-1)\,\ldots\,2\,(n-k)\,(n-k+1)\,\ldots\,n]$ \\
        
        $\vdots$&\qquad $\vdots$\\
        
        $\left(r_{n-k-1} r_{n-k-2}\right)^{n-k-4}$ & 
        \\
        
        $r_{n-k-2}$ & $[\ul{n-k-1}\,1\ldots\,(n-k-3)\,\ul{n-k-2}\,(n-k)
        \ldots\,n]$ \\
        
        $r_{n-k-1}$ & $[\ul{n-k-2}\,\ldots\,\ul{1}\,(n-k-1)\,(n-k)\,\ldots
        \,n]$ \\
        
        $r_{n-k-2}$ & $[1\,\ldots\,(n-k-2)\,(n-k-1)\,(n-k)\,\ldots
        \,n]$ \\
    \end{tabular}
    \caption{Base cycle $C_k$ when $\left\lceil n/2 \right\rceil \leq k < n-1$. If $n-k = 0,1,2,3,$ or $4$, the $(r_{n-k-1}r_{n-k-2})^{n-k-4}$ portion of the cycle described here is empty.}
    \label{tab:basecycle2}
\end{table}

\begin{table}
    \centering
    \begin{tabular}{l|l}
         \hline
        Edge label & Node \\ \hline
         & $[1\,2\,3\,\ldots\,\ldots\,(n-2)\,(n-1)\,n]$ \\
        $r_{n}$ & $[\ul{n}\,\ul{n-1}\,\ul{n-2}\,\ldots\,\ul{3}\,\ul{2}\,\ul{1}]$ \\
        $r_{n-1}$ & $[2\,3\,4\,\ldots\,(n-1)\,n\,\ul{1}]$ \\
        $r_{n}$ & $[1\,\ul{n}\,\ul{n-1}\,\ldots\,\ul{4}\,\ul{3}\,\ul{2}]$ \\
        $r_{n-1}$ & $[3\,4\,5\,\ldots\,n\,\ul{1}\,\ul{2}]$ \\
        $\vdots$&\qquad$\vdots$\\
          $\left(r_{n-1}r_{n}\right)^{n-2}$ & $[\ul{1}\,\ul{2}\,\ul{3}\,\ldots\,\ul{n-2}\,\ul{n-1}\,\ul{n}]$ \\
          $r_{n}$ & $[n\,(n-1)\,(n-2)\,\ldots\,2\,1]$ \\
          $r_{n-1}$ & $[\ul{2}\,\ul{3}\,\ul{4}\,\ldots\,\ul{n-1}\,\ul{n}\,1]$ \\
         $\vdots$&\qquad$\vdots$\\
          $\left(r_{n-1}r_{n}\right)^{n-2}$ & $[\ul{n}\,1\,2\,\ldots\,(n-3)\,(n-2)\,(n-1)]$ \\
        $r_{n}$ & $[\ul{n-1}\,\ul{n-2}\,\ul{n-3}\,\ldots\,\ul{2}\,\ul{1}\,n]$ \\
        $r_{n-1}$ & $[1\,2\,3\,\ldots\,\ldots\,(n-2)\,(n-1)\,n]$
    \end{tabular}
    \caption{Base cycle $C_{n}$, with $n\geq4$.}
    \label{tab:basecycle3}
\end{table}


We now make several observations regarding certain inequalities. These will be needed within the proof of Theorem~\ref{thm:main}. They can all be easily proven by induction or elementary algebraic manipulation.

\begin{obs}\label{o:ineq}
When 
$\left\lceil n/2 \right\rceil \leq k \leq n-1$,
then $n-k \leq k+1$.
\end{obs}

\begin{obs}\label{o:1ii1}
For $n \geq 5$, $2^{n-2}(n-2)! < 2^{n-1}(n-2)!-3n-9 < 2^{n-1}(n-1)!$. 
\end{obs}

\begin{obs}\label{o:1ii2} 
For $n \geq 3$, $2^{n-1}(n-2)!+2^{n-2}(n-2)!< 2^{n-1}(n-1)!$.
\end{obs}

\begin{obs}\label{o:1ii3}
For $n \geq 3$, $2^{n-1}(n-1)!-2^{n-1}(n-2)!> 2^{n-2}(n-2)!$.
\end{obs}

\begin{obs}\label{o:smallALowerBound}
For $n \geq 4$, $2^{n-2}(n-2)! < 2^{n-1}(n-1)!-17 < 2^{n-1}(n-1)!$.
\end{obs}

\begin{obs}\label{o:bigAUpperBound}
For $n \geq 3$, $2^{n-2}(n-2)! < 2^{n-2}(n-1)!-n-1 < 2^{n-1}(n-1)!$.
\end{obs}

\begin{obs}\label{o:bigALowerBound}
For $n \geq 3$, $2^{n-2}(n-2)! < 2^{n-1}(n-1)! - 4 < 2^{n-1}(n-1)!$.
\end{obs}

\begin{obs}\label{o:lowerBound3}
For $n \geq 4$, $2^{n-2}(n-2)! < 2^{n-2}(n-1)! - n + 1 < 2^{n-1}(n-1)!$.
\end{obs}

We are now ready to state and prove our first main result. 

\begin{thm}\label{thm:main}
For $BP_n$, with $n\geq2$, there exists a cycle of length $\ell$, with $8 \leq \ell \leq 2^n n!$, embedded in the graph. In particular, $BP_n$ has a Hamiltonian cycle.
\end{thm}

The idea of the proof can be summarized as follows: We use the fact that there are $2n$ copies of $BP_{n-1}$ embedded in $BP_n$, use the induction hypothesis and combine different cycles of length $k_i$, with $8\leq k_i\leq 2^{n-1}(n-1)!$ in $i$ of the copies of $BP_{n-1}$ with a base cycle, removing edges in common to obtain a larger cycle.  Inside the copies of $BP_{n-1}$, we use cycles having edges in $E^{n-1}_{n}$, and since the set $E^{n-1}_{n}$ is edge-transitive, we assume we can build cycles inside copies of $BP_{n-1}$ using the same edges labeled by $r_{n-1}$ in the base cycles. We make this process rigorous in the proof that follows. A schematic representation of the proof is given in Figure~\ref{f:proof}.

\begin{figure}[t]
\centering
\begin{tikzpicture}[scale=0.5,line cap=round,line join=round,>=triangle 45,x=1.0cm,y=1.0cm]
\draw [rotate around={45.:(-0.22,-1.24)}] (-0.22,-1.24) ellipse (1.7708575554888042cm and 1.5543283056779813cm);
\draw [rotate around={45.:(3.8715139033861563,2.7028644481044513)}] (3.8715139033861563,2.7028644481044513) ellipse (1.7708575554888022cm and 1.5543283056779795cm);
\draw [rotate around={45.:(9.8,4.6)}] (9.8,4.6) ellipse (1.7708575554888104cm and 1.5543283056779869cm);
\draw [rotate around={45.000000000000135:(14.94,-0.5)}] (14.94,-0.5) ellipse (1.7708575554888109cm and 1.554328305677988cm);
\draw [rotate around={45.:(7.0590794409271895,-5.5085417101584655)}] (7.0590794409271895,-5.5085417101584655) ellipse (1.7708575554888124cm and 1.5543283056779889cm);
\draw [thick] (-0.7264910708510319,-0.5963232353696936)-- (-0.6936802465963368,-1.9415670298121919);
\draw [thick] (2.9811320699295125,3.1769215539202413)-- (3.0139428941842086,1.8316777594777416);
\draw [thick] (9.149567029812188,5.014327712183163)-- (9.182377854066884,3.6690839177406662);
\draw [thick] (14.33367726205401,0.1255148982335983)-- (14.366488086308706,-1.2197288962089008);
\draw [thick] (7.344971695803958,-4.796108739970666)-- (7.377782520058654,-6.141352534413165);
\draw [thick] (-0.7240917061760453,-0.6946971870441451)-- (3.0123823317939733,1.8956608174773852);
\draw [thick] (2.9811905910191463,3.1745221892452538)-- (9.15038632506706,4.980736606733352);
\draw [thick] (9.181578065841888,3.7018752349654838)-- (14.336096133758875,0.026341158334151432);
\draw [thick,dash pattern=on 6pt off 6pt] (14.366488086308706,-1.2197288962089008)-- (7.377782520058654,-6.141352534413165);
\draw [thick] (7.346590779283827,-4.862491162645297)-- (-0.6936802465963368,-1.9415670298121919);
\filldraw (-0.7264910708510319,-0.7) circle (3pt);
\filldraw (2.9811320699295125,3.1769215539202413) circle (3pt);
\filldraw (9.149567029812188,5) circle (3pt);
\filldraw (14.33367726205401,0.095) circle (3pt);
\filldraw (7.344971695803958,-4.85) circle (3pt);
\filldraw (-0.7264910708510319,-1.9) circle (3pt);
\filldraw (2.9811320699295125,1.9) circle (3pt);
\filldraw (9.16,3.8) circle (3pt);
\filldraw (14.33367726205401,-1.2) circle (3pt);
\filldraw (7.344971695803958,-6.1) circle (3pt);
\draw [shift={(0.08441459317116995,-1.2987833140487204)},dash pattern=on 6pt off 6pt]  plot[domain=-2.4511387732914405:2.499909591636878,variable=\t]({1.*1.0092583836280562*cos(\t r)+0.*1.0092583836280562*sin(\t r)},{0.*1.0092583836280562*cos(\t r)+1.*1.0092583836280562*sin(\t r)});
\draw [shift={(3.824848558206411,2.540083123750604)},dash pattern=on 6pt off 6pt]  plot[domain=-2.4511387732914414:2.4999095916368783,variable=\t]({1.*1.0092583836280566*cos(\t r)+0.*1.0092583836280566*sin(\t r)},{0.*1.0092583836280566*cos(\t r)+1.*1.0092583836280566*sin(\t r)});
\draw [shift={(10.026094342343782,4.443110930522919)},dash pattern=on 6pt off 6pt]  plot[domain=-2.4511387732914445:2.4999095916368854,variable=\t]({1.*1.009258383628057*cos(\t r)+0.*1.009258383628057*sin(\t r)},{0.*1.009258383628057*cos(\t r)+1.*1.009258383628057*sin(\t r)});
\draw [shift={(15.144582926076213,-0.5441343561907394)},dash pattern=on 6pt off 6pt]  plot[domain=-2.451138773291442:2.499909591636876,variable=\t]({1.*1.009258383628055*cos(\t r)+0.*1.009258383628055*sin(\t r)},{0.*1.009258383628055*cos(\t r)+1.*1.009258383628055*sin(\t r)});
\draw [shift={(6.584023535299155,-5.6169331174300945)},dash pattern=on 6pt off 6pt]  plot[domain=0.6904538802983542:5.641502245226671,variable=\t]({1.*1.0092583836280542*cos(\t r)+0.*1.0092583836280542*sin(\t r)},{0.*1.0092583836280542*cos(\t r)+1.*1.0092583836280542*sin(\t r)});
\begin{scriptsize}
\draw[color=black] (3.5,5) node {$BP_{n-1}(n)$};
\draw[color=black] (0.1,-1.4) node {$r_{n-1}$};
\draw[color=black] (3.75,2.4) node {$r_{n-1}$};
\draw[color=black] (9.95,4.3) node {$r_{n-1}$};
\draw[color=black] (15.2,-0.5) node {$r_{n-1}$};
\draw[color=black] (6.5,-5.5) node {$r_{n-1}$};
\end{scriptsize}
\end{tikzpicture}
\caption{Depiction of the proof of Theorem~\ref{thm:main}. There is a base cycle that starts and ends with the copy of $BP_{n-1}$ of all the signed permutations that end with $n$, which we denote by $BP_{n-1}(n)$. The idea of the proof is to apply the induction hypothesis and replace the edge labeled by $r_{n-1}$ in certain copies of $BP_{n-1}$ by paths with the appropriate lengths, sketched in the picture by the dashed arch.}
\label{f:proof}
\end{figure}
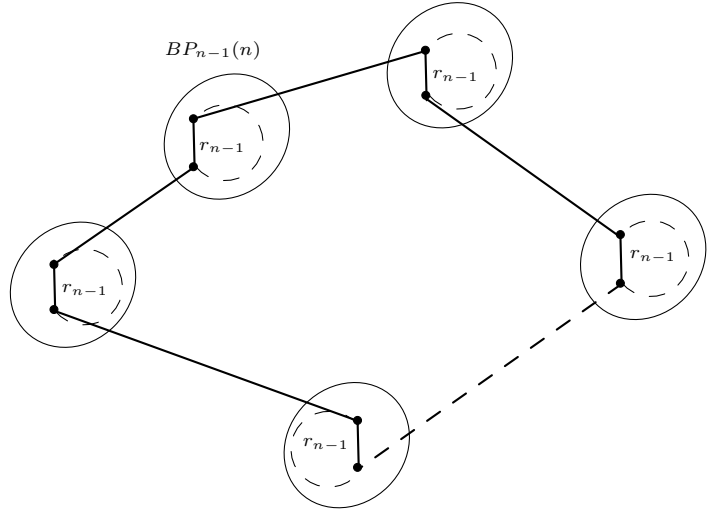

\textit{Proof of Theorem~\ref{thm:main}.} We proceed by induction on $n$. If $n=2$, then $BP_2$ is isomorphic to an $8$-cycle, and thus the theorem is true in this case. Moreover, if $n=3$, a computer search using Depth-First Search to detect cycles of a given length confirms the existence of cycles of all lengths from 8 to 48 (see the authors' preprinted version~\cite[Appendix A]{BBParxiv}). Similarly, if $n=4$, another computer search gives the existence of cycles of length 8 to 384 (see the authors' preprinted version~\cite[Appendix B]{BBParxiv}).  

Let us assume that $BP_{n-1}$ is Hamiltonian and weakly pancyclic, where $n\geq5$. It is worth note that any cycle of length $\ell$, where $2^{n-2}(n-2)!<\ell\leq2^{n-1}(n-1)!$, in $BP_{n-1}$ will necessarily include an edge labeled by $r_{n-1}$ since $\ell$ is greater than the length of any Hamiltonian cycle of any copy of $BP_{n-2}$ embedded in $BP_{n-1}$. By Lemma~\ref{l:edgetrans}, the set of edges labeled by $r_{n-1}$ is edge-transitive. So, any given cycle including an edge labeled by $r_{n-1}$ may be mapped to a cycle that includes any other edge labeled by $r_{n-1}$.

We now show that $BP_n$ is also Hamiltonian and weakly pancyclic. By the induction hypothesis, there are cycles of length from $8$ to $2^{n-1}(n-1)!$ within any embedded copy of $BP_{n-1}$ in $BP_n$. So, it remains to show that there are cycles of lengths between $2^{n-1}(n-1)!+1$ and $2^n n!$.

\textbf{Cycle of length $2^{n-1}(n-1)!+1$}. Begin with $C_1$ which has length $19$ and has an $r_{n-1}$-labeled edge in 5 copies of $BP_{n-1}$. The second copy of $BP_{n-1}$ visited by $C_1$ has a cycle $C$ of length $2^{n-1}(n-1)!-16$. Since $2^{n-1}(n-1)!-16>2^{n-2}(n-2)!$ as $n\geq4$, $C$ contains an edge labeled by $r_{n-1}$. Thus, we can assume that $C$ uses the edge labeled by $r_{n-1}$ in the second copy of $BP_{n-1}$ visited by $C_1$. Merge the cycles $C_1$ and $C$ removing the common edge to construct a new cycle $C'$. The length of $C'$ is $19-1+2^{n-1}(n-1)!-16-1=2^{n-1}(n-1)!+1$.

\textbf{Cycle of length $2^n n!$}. Begin with $C_{n}$, which has length $4n$ and has an $r_{n-1}$ edge in all $2n$ copies of $BP_{n-1}$ embedded in $BP_n$. By the induction hypothesis, there is a cycle of length $2^{n-1} (n-1)!$ in each of the copies of $BP_{n-1}$, and clearly each must include an $r_{n-1}$ edge. Remove these edges from each of these cycles and attach the resulting paths in place of the $r_{n-1}$ edges of $C_{n}$. The resulting cycle is of length $4n - 2n + 2n \left(2^{n-1} (n-1)!\right) - 2n = 2^n n!$.

\textbf{Cycles with lengths of the form $a\left(2^{n-1} (n-1)!\right)+b$} for $1 \leq a \leq 2n-1$ and $0 \leq b \leq 2^{n-1}(n-1)!-1$.

There are three cases for these cycles, which depend upon the size of $a$. In two of the cases, we start with the base cycle $C_a$.

    \textbf{Case 1, for $1 \leq a < \left\lceil n/2 \right\rceil$.} We further subdivide this case into two subcases depending on the value of $b$.
    
    \noindent \textbf{Subcase 1(i), for $b-6a-9>2^{n-2}(n-2)!$.}
    The base cycle $C_a$ has length $8a+11$ and visits $2a+3$ copies of $BP_{n-1}$. Of those $2a+3$ copies of $BP_{n-1}$, $2a+1$ include a single edge labeled by $r_{n-1}$. We replace $a+1$ of these edges labeled by $r_{n-1}$ as follows. In $a$ of these copies, make a Hamiltonian cycle $C^i$, for $1 \leq i \leq a$, of length $2^{n-1}(n-1)!$, which exists by the induction hypothesis. In a different copy of $BP_{n-1}$, make a cycle $C$ of length $b-6a-9$. Since $b-6a-9>2^{n-2}(n-2)!$ by assumption, $C$ uses an edge labeled by $r_{n-1}$, and since $E^{n-1}_n$ is edge-transitive, we assume that this edge is the same used in $C_a$. Merge each $C^i$ and $C$ with $C_a$, removing the edges $r_{n-1}$ in common to obtain a cycle $C'$. Notice that the length of $C'$ is \[
    8a+11-(a+1)+a\left( 2^{n-1}(n-1)! \right)-a+(b-6a-9)-1=a \left(2^{n-1}(n-1)! \right)+b,
    \]
    as desired. 
    
    \noindent \textbf{Subcase 1(ii), for $b-6a-9\leq 2^{n-2}(n-2)!$.} We start with $C_a$ and again, we shall replace $a+1$ edges in $C_a$ labeled by $r_{n-1}$ by appropriate cycles. In $a-1$ copies of $BP_{n-1}$, construct a Hamiltonian cycle $C^i$, for $1 \leq i \leq a-1$, of length $2^{n-1}(n-1)!$; in another copy of $BP_{n-1}$ construct a cycle $C'$ of length  $\ell'=2^{n-1}(n-2)!+b-6a-9$. By Observations \ref{o:1ii1} and \ref{o:1ii2}, $2^{n-2}(n-2)!<\ell'\leq2^{n-1}(n-1)!$. Therefore $C'$ exists and it must have an edge labeled by $r_{n-1}$ and, therefore, we shall assume that it is the same edge used by $C_a$. Finally, take one last copy of $BP_{n-1}$ with one occurrence of an edge labeled by $r_{n-1}$ also present in $C_a$, and construct a cycle $C''$ of length $\ell''=2^{n-1}(n-1)!-2^{n-1}(n-2)!$. Observation~\ref{o:1ii3} gives that $2^{n-2}(n-2)!<\ell''\leq 2^{n-1}(n-1)!$. Hence, $C''$ exists and we can assume that it uses the same edge labeled by $r_{n-1}$ that $C_a$ uses in the respective copy of $BP_{n-1}$. Merge each cycle $C^i, C'$, and $C''$ with $C_a$, removing the edges labeled by $r_{n-1}$ in common among all of the cycles. The length of the cycle obtained is
    \begin{align*}
    &8a+11-(a+1)+(a-1)\left(2^{n-1}(n-1)!\right) -(a-1)\\
    &\quad +\left(2^{n-1}(n-2)!+b-6a-9\right)-1+\left(2^{n-1}(n-1)!-2^{n-1}(n-2)!\right)-1\\
    &= a\left(2^{n-1}(n-1)!\right)+b,
    \end{align*}
    as desired. 
  
    \textbf{Case 2, for $\left\lceil n/2 \right\rceil \leq a \leq n-2$}. The base cycle $C_a$ has length $4n+2$ and visits $2a+4$ copies of $BP_{n-1}$. Of those $2a+4$ copies of $BP_{n-1}$,  $2a+2$ only include a single $r_{n-1}$ edge. We replace $2a$ of these edges in $C_a$ labeled by $r_{n-1}$ by cycles of appropriate length in $BP_{n-1}$ as follows. In $2a-2$ of these copies, make a cycle $C^i$, for $1 \leq i \leq 2a-2$, of length $2^{n-2}(n-1)!$, which exists by the induction hypothesis. Moreover, since $2^{n-2}(n-1)!>2^{n-2}(n-2)!$, each $C^i$ must use an edge labeled by $r_{n-1}$. Since $E^{n-1}_n$ is edge-transitive, we shall assume each $C^i$ uses the edges labeled by $r_{n-1}$ from $C_a$. In two of the other copies with only a single edge of $C_a$, attach two cycles $C'$ and $C''$. The cycle $C'$ is of length $2^{n-2}(n-1)!-2n+2a+\left\lfloor b/2 \right\rfloor$, and by the bounds of $\left\lceil n/2 \right\rceil \leq a$ and $0 \leq b$, is at least of length $2^{n-2}(n-1)!-2n+(n-1) = 2^{n-2}(n-1)!-n-1$, which by Observation \ref{o:bigALowerBound} is greater than $2^{n-2}(n-2)!$. The cycle $C''$ is of length $2^{n-2}(n-1)!-2n+2a+\left\lceil b/2 \right\rceil$, and by the bounds of $a \leq n-1$ and $b < 2^{n-1}(n-1)!$, is at most of length $2^{n-1}(n-1)!-4$, which by Observation \ref{o:bigAUpperBound} is less than $2^{n-1}(n-1)!$. Thus, $C'$ and $C''$ are cycles that will contain an $r_{n-1}$ edge. Merge each cycle $C^i, C'$, and $C''$ with $C_a$, removing the edges labeled by $r_{n-1}$ in common among all of the cycles. The resulting cycle is of length 
    \begin{align*} 
        &4n+2 - 2a + (2a-2) \left(2^{n-2}(n-1)!\right) - (2a-2) \\ 
        & \quad + \left(2^{n-2}(n-1)!-2n+2a+\left\lfloor b/2 \right\rfloor\right) - 1 \\ 
        & \quad + \left(2^{n-2}(n-1)!-2n+2a+\left\lceil b/2 \right\rceil\right) - 1 \\ &= a\left(2^{n-1}(n-1)!\right)+b.
    \end{align*}
    
    \textbf{Case 3, for $n-1 \leq a \leq 2n-1$}. The base cycle $C_n$ has length $4n$ and visits all $2n$ copies of $BP_{n-1}$. In each copy, only an edge labeled by $r_{n-1}$ is used. In $a-1$ of the copies, make a cycle $C^i$, for $1 \leq i \leq a-1$, of length $2^{n-1}(n-1)!$, which exists by the induction hypothesis and contains an edge labeled by $r_{n-1}$. In two of the remaining copies, make two cycles $C'$ and $C''$, whose lengths are as follows:  The cycle $C'$ is of length $2^{n-2}(n-1)!-2n+a+\left\lfloor b/2 \right\rfloor + 1$, and due to $b \geq 0$ along with $a \geq n$ is at least of length $2^{n-2}(n-1)!-n+1$, which by Observation \ref{o:lowerBound3} is greater than $2^{n-2}(n-2)!$. The cycle $C''$ is of length $2^{n-2}(n-1)!-2n+a+\left\lceil b/2 \right\rceil+1$, and due to $b < 2^{n-1}(n-1)!$ as well as $a\leq 2n-1$, is at most of length $2^{n-1}(n-1)!$, which is the length of the longest cycle in a copy of $BP_{n-1}$. Thus, both $C'$ and $C''$ are cycles containing an $r_{n-1}$ edge. Since $E^{n-1}_n$ is edge-transitive, each of the edges labeled by $r_{n-1}$ in each of the $C^i,C',$ and $C''$ cycles may be assumed to be the edges labeled by $r_{n-1}$ in each of the copies of $BP_{n-1}$ that $C_n$ goes through. Merge each $C^i, C',$ and $C''$ with $C_n$ by removing the $r_{n-1}$ edges in common to each of them. The resulting cycle is of length
    \begin{align*}
        &4n - (a+1) + (a-1)\left(2^{n-1}(n-1)!\right) - (a-1) \\
        &\quad + \left(2^{n-2}(n-1)!-2n+a+\left\lfloor b/2 \right\rfloor + 1\right) - 1 \\
        &\quad + \left(2^{n-2}(n-1)! -2n + a+\left\lceil b/2 \right\rceil + 1\right) - 1\\
        &= a\left(2^{n-1}(n-1)!\right)+b.
    \end{align*}
    
    All of these cycle lengths encompass all lengths between $8$ and $2^n n!$. Therefore, $BP_n$ is Hamiltonian and weakly pancyclic. \hfill $\square$

\section{Classification of $8$-cycles in $BP_n$, for $ n\geq2$}\label{s:classification}

In this section, we provide the classification of the $8$-cycles of $BP_n$, for $n\geq2$. These are the shortest possible cycles that can be embedded in the burnt pancake graph, by Theorem \ref{thm:main}. This constitutes the beginning of a process of describing all the small cycles in $BP_n$, inspired by the work of classifying the 6, 7, 8, and 9-cycles of $P_n$ carried out in Konstantinova and Medvedev~\cite{KM10, KM11, KonMed}. Just like the previous section, we take advantage of the recursive structure of $BP_n$. We use the notation $BP_{n-1}(q)$ to denote the copy of $BP_{n-1}$ obtained by restricting $BP_n$ to its subgraph induced by all of those signed permutations whose last character is $q$ in the window notation, where $q\in[\pm n]$.

Throughout this discussion, we consider decomposing the window notation into substrings and track their orientation and position after pancake flips. We adopt the convention of using a capital letter to represent substrings. For example, for $\pi \in B_n$, we may write $\pi = [XYZ]$. The lengths of these substrings will be of importance and to be succinct we use the symbol $|X|$ for the length of the substring. For convenience, the reversal and sign change of a substring will be denoted by $\overline{X}$. That is, if $X=[x_1\;x_2\;\cdots\;x_{i-1}\;x_i]$, then  $\overline{X}=[\ul{x_i}\;\ul{x_{i-1}}\;\cdots\;\ul{x_2}\;\ul{x_1}]$. Throughout this section we also adopt the convention of using names of signed permutations based on the last character, e.g. $\pi \in V(BP_{n-1}(p))$, $\ul{\pi} \in V(BP_{n-1}(\ul{p}))$, $\rho \in V(BP_{n-1}(q))$, etc. For $\pi,\tau \in B_n$, we define the \emph{distance} $d(\pi, \tau)$ as the length of the shortest $\pi$-$\tau$ path in $BP_n$.

Following the definition in Konstantinova and Medvedev~\cite{KM10, KM11, KonMed}, we say that a cycle $C=r_{i_1}r_{i_2}\cdots r_{i_k}$ is in \emph{canonical form} if the index sequence $i_1,i_2,\ldots,i_k$ is lexicographically maximal.

\begin{thm}\label{thm:classification}
An $8$-cycle in $BP_n$, with $n\geq2$, has one of the following canonical forms:
\end{thm}

\begin{align}
    \label{8-1}
    &r_{k}r_{j}r_{i}r_{j}r_{k}r_{k-j+i}r_{i}r_{k-j+i}, & \text{ for }1 \leq i < j \leq k-1,  \text{ and } 3 \leq k \leq n\\
    \label{8-2}
    &r_{k}r_{j}r_{k}r_{i}r_{k}r_{j}r_{k}r_{i}, & \text{ for }2 \leq i,j \leq k-2, i+j \leq k,  \text{ and } 4 \leq k \leq n\\
    \label{8-3}
    &r_{k}r_{i}r_{k}r_{1}r_{k}r_{i}r_{k}r_{1}, &\text{ for }2 \leq i \leq k-1,  \text{ and } 3 \leq k \leq n\\
    \label{8-4}
    &r_{k}r_{1}r_{k}r_{1}r_{k}r_{1}r_{k}r_{1}, & \text{ for }2 \leq k \leq n
\end{align}

Before proving Theorem~\ref{thm:classification}, we need a few preliminary lemmas. 

\begin{lem}\label{lem:31}
 If $\pi \in V(BP_{n-1}(p))$ and $\pi r_{n} \in V(BP_{n-1}(q))$, then $|p| \neq |q|$.
\end{lem}
\begin{proof} 
We proceed by proving the contrapositive statement. Suppose $|p| = |q|$. Since $p\neq q$, then $q=\ul{p}$. Let $\pi \in V(BP_{n-1}(p))$ such that $\pi r_{n} \in V(BP_{n-1}(\ul{p}))$. So, the first and last element of $\pi$ in the window notation must be $p$, which is impossible since then $p$ would appear twice in the window notation of $\pi$.
\end{proof}

\begin{lem}\label{lem:32}
Let $\pi,\tau\in B_n$ have the same first element $q\in[\pm n]$ in the window notation. Then $d(\pi,\tau) = 3$ if and only if $\tau = \pi r_{j} r_{i} r_{j}$, for $1 \leq i < j \leq n$, where $\pi = [A B C]$, $\tau=[A \overline{B} C]$, $|A|=j-i$, $|B|=i$, and $|C| \geq 0$.
\end{lem}

\begin{proof}
Suppose $d(\pi,\tau)=3$, in $BP_n$ where both of their window notations begin with $q \in [\pm n]$. Decompose $\pi=[q A_1 A_2]$ with $|A_1|=j-1$, where $2 \leq j \leq n$. Since $\tau$ has $q$ at first position, $q$ must be involved in an even number of flips in the path from $\pi$ to $\tau$ to retain the sign of $q$. Since the path is of length $3$ between $\pi$ and $\tau$, then $q$ must be involved in $2$ flips. Let us find vertices with $q$ in the first position and that are at a distance of $3$ from $\pi$. Let $\pi_1 = \pi 
r_{j}
= [\overline{A_1}\uline{q}A_2]$. Now, to get the correct sign, $q$ can be used in only one more flip, and to place it in the first position this will need to be the $3\text{rd}$ flip. So the next flip must only include a substring of $A_1$, $\pi_2 = \pi_1 r_{i}
= [A_{12} \overline{A_{11}} \uline{q} A_2]$ where $A_1 = A_{11}A_{12}$, for $1 \leq i < j\leq n$, and $|A_{12}|=i$. Hence, $\tau = \pi r_{j}r_{i}r_{j}
= [q A_{11} \overline{A_{12}} A_2]$, and we get $\pi = [A B C]$ and $\tau=[A \overline{B} C]$, where $A=q A_{11}$, $B=A_{12}$, and $C=A_2$. Note that using $r_{j}$
is the only way to restore $q$ to the first position after reaching $\pi_2$. The converse is easily checked to hold. 
\end{proof}

\begin{lem}\label{lem:33}
If $\pi,\tau \in B_n$ with $\pi=[ABC]$, $\tau=[BAC]$, and $|A|,|B|,|C| \geq 1$, then $d\left(\pi,\tau\right) \geq 3$ in $BP_n$.
\end{lem}

\begin{proof}
Let $\pi=[ABC]$ and $\tau=[BAC]$. It is true that $\tau = \pi r_{|A|}r_{|A|+|B|}r_{|B|}$
is a path of length $3$ between the two signed permutations. If there was another path of length $1$ or $2$ between $\pi$ and $\tau$, then there would be cycles of length either $4$ or $5$, which do not exist in $BP_n$ since its girth is $8$. Thus, any path from $\pi$ to $\tau$ must have a length of at least $3$.
\end{proof}

\begin{lem}\label{lem:34}
If $\pi_1, \pi_2 \in V(BP_{n-1}(p))$, for any $p\in[\pm n]$, with $d(\pi_1,\pi_2) \leq 2$, then $\pi_1 r_{n}$ and $\pi_2 r_{n}$
must belong to distinct copies of $BP_{n-1}$ in $BP_{n}$.
\end{lem}

\begin{proof}
 Let $\pi_1, \pi_2 \in V(BP_{n-1}(p))$ for some $p\in[\pm n]$ and $d(\pi_1,\pi_2) \leq 2$. Let $\pi_1=[\uline{q} A p]$ be such that $\pi_1 r_n \in V(BP_{n-1}(q))$. Then, we have two cases to consider. 
\begin{description}
    \item[$d(\pi_1,\pi_2)=1$.] In this case, 
    $\pi_2 = \pi_1 r_{|A_1|+1} = [\overline{A_1} q A_2 p]$, where $A=A_1 A_2$, $|A_1| \geq 0$, and thus $\pi_2 r_{n}=[\uline{p} \overline{A_2} \uline{q} A_1] \notin BP_{n-1}(q)$. 
    \item[$d(\pi_1,\pi_2)=2$.]
    Let $\pi'$ be the vertex between $\pi_1$ and $\pi_2$. Then, $\pi' = \pi_1 r_{|A_1|+1}
    = [\overline{A_1} q A_2 p]$, where $A=A_1 A_2$ and $|A_1| \geq 0$. Then, $\pi_2 = \pi' r_{|A_{12}|}
    =[A_{12} \overline{A_{11}} q A_2 p]$, where $A_1=A_{11}A_{12}$, $|A_{12}| \geq 0$  or $\pi_2= \pi' r_{|A_1|+|A_{21}|+1}
    =[\overline{A_{21}}\uline{q}A_1 A_{22} p]$, where $A_2=A_{21}A_{22}$, $|A_{21}| \geq 1$. Therefore, it must be that $\pi_2 r_{n} = [\uline{p}\overline{A_2}\uline{q}A_{11}\overline{A_{12}}] \notin BP_{n-1}(q)$ or $
    \pi_2 r_{n} =[\uline{p}\overline{A_{22}} \overline{A_1}q A_{21}] \notin BP_{n-1}(q)$.
\end{description}

Hence, if $\pi_1,\pi_2 \in BP_{n-1}(p)$ are at a distance of at most $2$, then $\pi_1 r_{n}$ and $\pi_2 r_{n}$ belong to distinct copies of $BP_{n-1}$ in $BP_n$.
\end{proof}

We now proceed to the proof of Theorem~\ref{thm:classification}.

\noindent\textit{Proof of Theorem~\ref{thm:classification}.}
We use the recursive structure of $BP_n$. Notice that the edges labeled by $r_i,$ for $1\leq i\leq k-1$, are all contained inside a copy of $BP_{k-1}$, and only edges labeled by 
$r_{k}$ connect the different copies of $BP_{k-1}$ in $BP_n$. This observation allows us to obtain canonical forms for the $8$-cycles of $BP_n$ based on how many $BP_{k-1}$ copies are used by the $8$-cycle. To contain a non-trivial copy of $BP_{k-1}$ that is not a single point, we must have $k \geq 2$. It follows, then, that an $8$-cycle can be formed by using $2$, $3$, or $4$ copies of $BP_{k-1}$. To help us enumerate the cases, we use an integer partition $(a_1+a_2+\cdots+a_m)$ of 8 to indicate that the $8$-cycle visits $m$ copies of $BP_{k-1}$, and has $a_1$ vertices in one copy, $a_2$ vertices in a second copy, and so on, up to $a_m$ vertices in the $m$-th copy. Since no vertex meets two $r_{k}$ edges, then we must have $a_i > 1$, for $1 \leq i \leq m$. \;

\textbf{Case the $8$-cycle visits two copies of $BP_{k-1}$.}
Part sizes of the cycle partition are at least $2$, which leads to the only possible partitions being (6+2), (5+3), and (4+4). Suppose that the two copies met are $BP_{k-1}(p)$ and $BP_{k-1}(q)$.

\noindent \textbf{Subcases (6+2) and (5+3).} By Lemma~\ref{lem:34}, if the endpoints in the $BP_{k-1}(p)$ copy (say $\pi_1$ and $\pi_2$) are at a distance of at most $2$, then $\pi_1 r_{k}$ and $\pi_2 r_{k}$ will belong to distinct copies of $BP_{k-1}$. Hence, an $8$-cycle cannot occur with associated partition $(6+2)$ or $(5+3)$. Therefore, an $8$-cycle can be formed only when $4$ vertices belong to each copy.

\begin{figure}
    \centering
\begin{tikzpicture}
\draw (2,2) ellipse (1cm and 2cm);
\draw (6,2) ellipse (1cm and 2cm);
\draw node at (0,2) {$BP_{k-1}(p)$};
\draw node at (8,2) {$BP_{k-1}(q)$};
\filldraw (2.25,3.50) circle (2pt) node[align=left,xshift = -0.5cm] {$\pi_1$};
\filldraw (1.75,2.50) circle (2pt) node[align=left,xshift = -0.5cm] {$\pi_2$};
\draw [thick] (2.25,3.50) -- (1.75,2.50);
\filldraw (1.75,1.50) circle (2pt) node[align=left,xshift = -0.5cm] {$\pi_3$};
\draw [thick] (1.75,2.50) -- (1.75,1.50);
\filldraw (2.25,0.50) circle (2pt) node[align=left,xshift = -0.5cm] {$\pi_4$};
\draw [thick] (1.75,1.50) -- (2.25,0.50);
\filldraw (5.75,0.50) circle (2pt) node[align=left,xshift = 0.5cm] {$\rho_4$};
\draw [thick] (2.25,0.50) -- (5.75,0.50);
\filldraw (6.25,1.50) circle (2pt) node[align=left,xshift = 0.5cm] {$\rho_3$};
\draw [thick] (5.75,0.50) --(6.25,1.50);
\filldraw (6.25,2.50) circle (2pt) node[align=left,xshift = 0.5cm] {$\rho_2$};
\draw [thick] (6.25,1.50) -- (6.25,2.50);
\filldraw (5.75,3.50) circle (2pt) node[align=left,xshift = 0.5cm] {$\rho_1$};
\draw [thick] (6.25,2.50) -- (5.75,3.50);
\draw [thick] (2.25,3.50) -- (5.75,3.50);
\end{tikzpicture}
    \caption{Subcase (4+4).}
    \label{f:2copies44}
\end{figure}
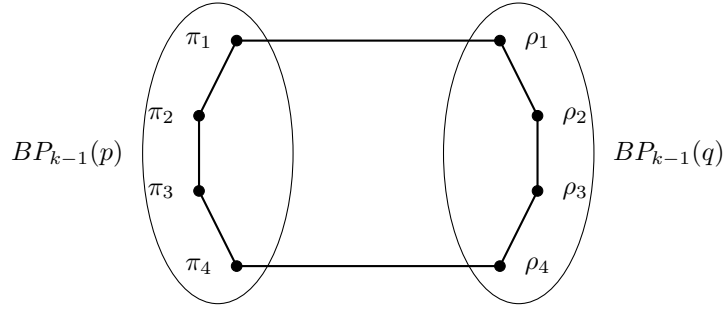

\noindent \textbf{Subcase (4+4).} By Lemma~\ref{lem:31}, $|p| \neq |q|$. Say the four vertices of $V(BP_{k-1}(p))$ are labeled by $\pi_i$, for $i \in [4]$, and the four vertices of $V(BP_{k-1}(q))$ are labeled by $\rho_i$, for $i \in [4]$. Each of $\pi_i$ and $\rho_i$ form a path of length three whose endpoints should be adjacent. We shall say that $\pi_1$ is adjacent to $\rho_1$ and $\pi_4$ is adjacent to $\rho_4$ (see Figure~\ref{f:2copies44}). This means both endpoints in $BP_{k-1}(p)$ ($\pi_1$ and $\pi_4$) should have $\ul{q}$ in their first positions and both endpoints in $BP_{k-1}(q)$ ($\rho_1$ and $\rho_4$) should have $\ul{p}$ in their first positions (see Figure~\ref{f:2copies44}). So in the window notation $\pi_1=[\uline{q}X p]$. By Lemma~\ref{lem:32}, we can get the forms of the remaining vertices of the cycle in $BP_{k-1}(p)$. So, $\pi_2=\pi_1 r_{j}
=[\overline{X_1}qX_2 p]$, $\pi_3=\pi_2 r_{i}
=[X_{12}\overline{X_{11}}qX_2 p]$, $\pi_4=\pi_3 r_{j}
=[\uline{q} X_{11} \overline{X_{12}}X_2 p]$, where $X=X_1X_2$, $X_1 = X_{11}X_{12}$, $j=|X_1|+1$, and $i = |X_{12}|$. After this, we get $\rho_1 = [\uline{p}\overline{X_2} \ \overline{X_{12}} \ \overline{X_{11}}q]$ and $\rho_4 = [\uline{p} \overline{X_2} X_{12} \overline{X_{11}}q]$. Taking $A=\uline{p}X_2$, $B=\overline{X_{11}}q$, $|A|,|B|,|X_{12}| \geq 1$, we need a path of length $3$ from $[AX_{12}B]$ and $[A\overline{X_{12}}B]$. $[AX_{12}B] \ r_{|A|+|X_{12}|}r_{|X_{12}|}r_{|A|+|X_{12}|} = [A\overline{X_{12}}B]$ is a path of length $3$, and there is no other path of length $3$ because there is no cycle of length 6 in the burnt pancake graph since its girth is 8. We have $|X_{11}| = j-1-i \geq 0, |A| = k-j \geq 1,$ and $|X_{12}|=i \geq 1$, which gives the cycle corresponding to the form (\ref{8-1}).

\textbf{Case the $8$-cycle visits three copies of $BP_{k-1}$.} Since part sizes are at least 2 in the partition, there are only two possibilities (4+2+2) or (3+3+2). Let the three copies used be $BP_{k-1}(p)$, $BP_{k-1}(q)$ and $BP_{k-1}(s)$. By Lemma~\ref{lem:31}, $|p|\neq |q|$ , $|p| \neq |s|$ and $|q| \neq |s|$.

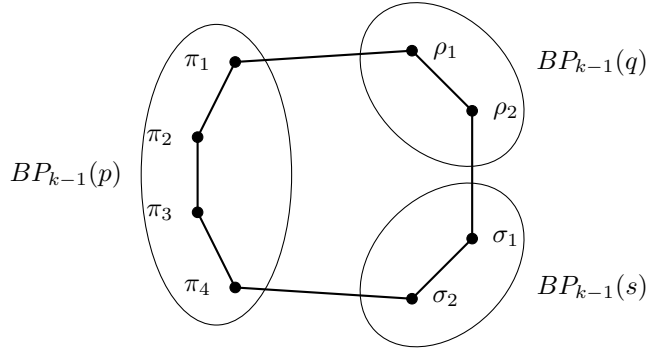
\begin{figure}
\centering
\begin{tikzpicture}
\draw (2,2) ellipse (1cm and 2cm);
\draw node at (0,2) {$BP_{k-1}(p)$};
\draw [rotate around={45:(5,3.2)}](5,3.2) ellipse (0.9cm and 1.25cm);
\draw node at (7,3.5) {$BP_{k-1}(q)$};
\draw [rotate around={-45:(5,0.8)}](5,0.8) ellipse (0.9cm and 1.25cm);
\draw node at (7,0.5) {$BP_{k-1}(s)$};
\filldraw (2.25,3.50) circle (2pt) node[align=left,xshift = -0.5cm] {$\pi_1$};
\filldraw (1.75,2.50) circle (2pt) node[align=left,xshift = -0.5cm] {$\pi_2$};
\draw [thick] (2.25,3.50) -- (1.75,2.50);
\filldraw (1.75,1.50) circle (2pt) node[align=left,xshift = -0.5cm] {$\pi_3$};
\draw [thick] (1.75,2.50) -- (1.75,1.50);
\filldraw (2.25,0.50) circle (2pt) node[align=left,xshift = -0.5cm] {$\pi_4$};
\draw [thick] (1.75,1.50) -- (2.25,0.50);
\filldraw (4.6,0.35) circle (2pt) node[align=left,xshift = 0.45cm] {$\sigma_2$};
\draw [thick] (2.25,0.50) -- (4.6,0.35);
\filldraw (5.4,1.15) circle (2pt) node[align=left,xshift = 0.45cm] {$\sigma_1$};
\draw [thick] (4.6,0.35) -- (5.4,1.15);
\filldraw (5.4,2.85) circle (2pt) node[align=left,xshift = 0.45cm] {$\rho_2$};
\draw [thick] (5.4,1.15) -- (5.4,2.85);
\filldraw (4.6,3.65) circle (2pt) node[align=left,xshift = 0.45cm] {$\rho_1$};
\draw [thick] (5.4,2.85) -- (4.6,3.65);
\draw [thick] (4.6,3.65) -- (2.25,3.50);
\end{tikzpicture}
 \caption{Subcase (4+2+2).}
    \label{f:3copies422}
\end{figure}

\noindent \textbf{Subcase (4+2+2).}
Suppose the four vertices in $BP_{k-1}(p)$ are labeled by $\pi_i$, for $i\in[4]$; the two vertices in $BP_{k-1}(q)$ are labeled by $\rho_i$, for $i\in[2]$; and the two vertices in $BP_{k-1}(s)$ are labeled by $\sigma_i$, for $i\in[2]$ (Figure~\ref{f:3copies422}). The vertex of $BP_{k-1}(p)$ that is adjacent to $BP_{k-1}(q)$ will have the window notation $\pi_1=[\uline{q} A s B p]$ or $\pi_1=[\uline{q} A \uline{s} B p]$. Now we consider the forms of the adjacent vertices in $BP_{k-1}(q)$ and $BP_{k-1}(s)$. 

If $\pi_1=[\uline{q} X s Y p]$, then $\rho_1=[\uline{p} \overline{Y} \uline{s} \overline{X} q]$. In order to be adjacent to $BP_{k-1}(s)$, $\rho_2$ must have $\uline{s}$ in its first position which is not possible by one flip from $\rho_1$ because any single flip changes the sign of any affected element.

If $\pi_1=[\uline{q} X \uline{s} Y p]$, then $\rho_1=[\uline{p} \overline{Y} s \overline{X} q]$. In order to be adjacent to $BP_{k-1}(s)$, it follows that $\rho_2 = [\uline{s} Y p \overline{X} q]$, and thus $\sigma_1 = [\uline{q} X \uline{p} \overline{Y} s]$. However, for $\sigma_2$ to be adjacent to $BP_{k-1}(p)$ it must have $\uline{p}$ in its first position, which is not possible by one flip from $\sigma_1$.

So cycles of the form $(4+2+2)$ do not exist in $BP_{n}$.

\begin{figure}
\centering
\begin{tikzpicture}
\draw (2,2) ellipse (0.85cm and 1.25cm);
\draw node at (0,2) {$BP_{k-1}(p)$};
\draw [rotate around={45:(5,3.7)}](5,3.7) ellipse (1.05cm and 1.75cm);
\draw node at (7,4.15) {$BP_{k-1}(q)$};
\draw [rotate around={-45:(5,0.3)}](5,0.3) ellipse (1.05cm and 1.75cm);
\draw node at (7,-0.15) {$BP_{k-1}(s)$};

\filldraw (2,2.50) circle (2pt) node[align=left,xshift = -0.5cm] {$\pi_1$};
\filldraw (2,1.50) circle (2pt) node[align=left,xshift = -0.5cm] {$\pi_2$};
\draw [thick] (2,2.50) -- (2,1.50);
\filldraw (4.2,-0.50) circle (2pt) node[align=left,xshift = 0.5cm,yshift =-0.2cm] {$\sigma_1$};
\draw [thick] (2,1.50) -- (4.2,-0.50);
\filldraw (5.2,0.1) circle (2pt) node[align=left,xshift = 0.5cm] {$\sigma_2$};
\draw [thick] (4.2,-0.50) -- (5.2,0.1);
\filldraw (5.7,1.10) circle (2pt) node[align=left,xshift = 0.45cm] {$\sigma_3$};
\draw [thick] (5.2,0.1) -- (5.7,1.10);
\filldraw (5.7,3) circle (2pt) node[align=left,xshift = 0.45cm] {$\rho_3$};
\draw [thick] (5.7,1.10) -- (5.7,3.0);
\filldraw (5.35,4.0) circle (2pt) node[align=left,xshift = 0.45cm] {$\rho_2$};
\draw [thick] (5.7,3.0) -- (5.35,4.0);
\filldraw (4.2,4.45) circle (2pt) node[align=left,xshift = 0.35cm,yshift = 0.2cm] {$\rho_1$};
\draw [thick] (5.35,4.0) -- (4.2,4.45);
\draw [thick] (4.2,4.45) -- (2,2.50);
\end{tikzpicture}
\caption{Subcase (3+2+2).}
    \label{f:3copies332}
\end{figure}
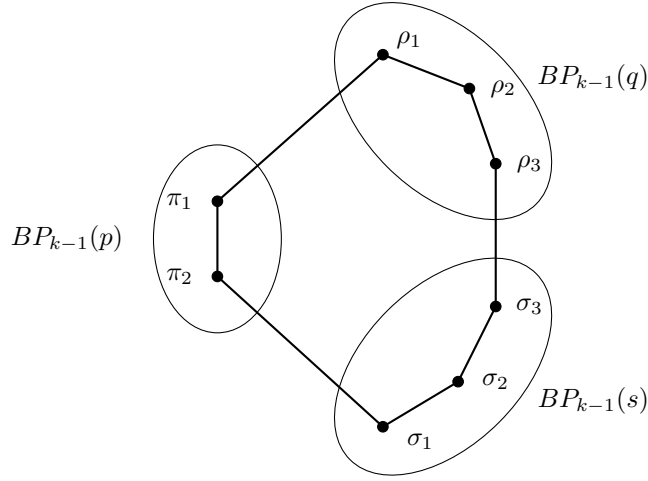

\noindent \textbf{Subcase (3+3+2).}
Suppose the three vertices in $BP_{k-1}(s)$ are labeled by $\sigma_i$, for $i\in[3]$; the three vertices in $BP_{k-1}(q)$ are labeled by $\rho_i$, for $i\in[3]$; and the two vertices in $BP_{k-1}(p)$ are labeled by $\pi_i$, for $i\in[2]$ (Figure~\ref{f:3copies332}). The vertex of $BP_{k-1}(p)$ that is adjacent to $BP_{k-1}(q)$ can have the form $\pi_1=[\uline{q} X s Y p]$ or $\pi_1=[\uline{q} X \uline{s} Y p]$.

If $\pi_1=[\uline{q} X s Y p]$, then $\rho_1=[\uline{p} \overline{Y} \uline{s} \overline{X} q]$, $\pi_2 = [\uline{s} \overline{X} q Y p]$, and $\sigma_1 = [\uline{p} \overline{Y} \uline{q} X s]$. Now we need to find a path of length $2$ between $\sigma_1$ and $\sigma_3$ that places $\ul{q}$ in the first position of $\sigma_3$. In order for the sign to be correct each of the two flips must affect $q$. Thus, $\sigma_2 = \sigma_1 r_{|Y|+|X_1|+2}
=[\overline{X_1} q Y p X_2 s]$, where $X = X_1 X_2$ and $|X_1| \geq 0$. This gives $\sigma_3=\sigma_2 r_{|X_1|+1}
=[\uline{q} X_1 Y p X_2 s]$ and $\rho_3=[\uline{s} \overline{X_2} \uline{p} \overline{Y} \ \overline{X_1} q]$. We now need a path of length 2 between $\rho_3$ and $\rho_1$. Taking $A = \uline{s}\overline{X_2}$, $B=\uline{p} \overline{Y}$ and $C=\overline{X_1 q}$, then $\rho_3 = [ABC]$ and $\rho_1 = [BAC]$. Thus, by Lemma~\ref{lem:33} we know that such a path between $\rho_3$ and $\rho_1$ does not exist.

If $\pi_1=[\uline{q} X \uline{s} Y p]$, then $\rho_1=[\uline{p} \overline{Y} s \overline{X} q]$. $\pi_2$ must have $\uline{s}$ at the first position, which is not possible by one flip from $\pi_1$. Thus, $\pi_1$ could not be of this form and be part of such an $8$-cycle.

So cycles of the form $(3+3+2)$ do not exist in $BP_{n}$. Furthermore, there are no $8$-cycles that visit $3$ copies of $BP_{k-1}$ in $BP_n$.

\textbf{Case the $8$-cycle visits four copies of $BP_{k-1}$ with (2+2+2+2).}
As stated earlier, there must be at least $2$ vertices in any copy of $BP_{k-1}$. Thus the only partition is (2+2+2+2). Let the four copies used be $BP_{k-1}(p)$, $BP_{k-1}(q)$, $BP_{k-1}(s)$ and $BP_{k-1}(t)$. The absolute values of $p,q,s,$ and $t$ may not be distinct. However, by Lemma~\ref{lem:31}, only non-adjacent copies can have the same absolute value in the last elements. Suppose the vertices in $BP_{k-1}(p)$ are labeled by $\pi_i$, for $i\in[2]$; the vertices in $BP_{k-1}(q)$ are labeled by $\rho_i$, for $i\in[2]$; the vertices in $BP_{k-1}(s)$ are labeled by $\sigma_i$, for $i\in[2]$; and the vertices in $BP_{k-1}(t)$ are labeled by $\tau_i$, for $i\in[2]$.

This gives rise to three subcases.

\noindent \textbf{Subcase the absolute value of $p,q,s$, and $t$ are pairwise different.} Suppose, without loss of generality, that $BP_{k-1}(p)$ and $BP_{k-1}(s)$ are non-adjacent (See Figure~\ref{f:4copies2222-1}). 

        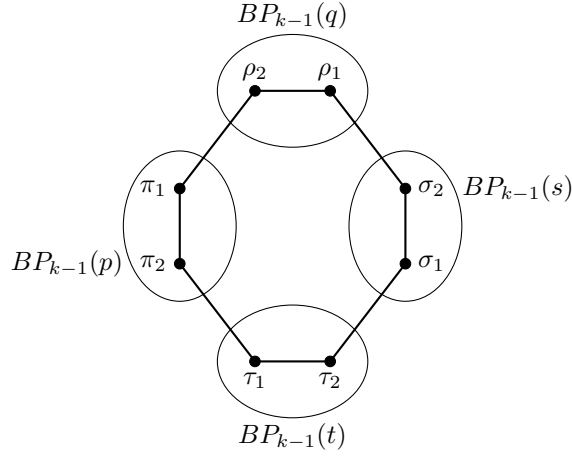
\begin{figure}
        \centering
        \begin{tikzpicture}[trim left=-4.cm]
        \draw (-2,2) ellipse (0.75cm and 1cm);
        \draw node at (-3.5,1.5) {$BP_{k-1}(p)$};
        \draw (1,2) ellipse (0.75cm and 1cm);
        \draw node at (-0.5,4.8) {$BP_{k-1}(q)$};
        \draw (-0.5,3.8) ellipse (1cm and 0.75cm);
        \draw node at (2.5,2.5) {$BP_{k-1}(s)$};
        \draw (-0.5,0.2) ellipse (1cm and 0.75cm);
        \draw node at (-0.5,-0.8) {$BP_{k-1}(t)$};
        \filldraw (-2,2.50) circle (2pt) node[align=left,xshift = -0.35cm] {$\pi_1$};
        \filldraw (-2,1.50) circle (2pt) node[align=left,xshift = -0.35cm] {$\pi_2$};
        \draw [thick] (-2,2.50) -- (-2,1.50);
        \filldraw (-1,0.2) circle (2pt) node[align=left,yshift = -0.25cm] {$\tau_1$};
        \draw [thick] (-2,1.50) -- (-1,0.20);
        \filldraw (0,0.2) circle (2pt) node[align=left,yshift = -0.25cm] {$\tau_2$};
        \draw [thick] (-1,0.20) -- (0,0.2);
        \filldraw (1,1.5) circle (2pt) node[align=left,xshift = 0.35cm] {$\sigma_1$};
        \draw [thick] (0,0.2) -- (1,1.50);
        \filldraw (1,2.5) circle (2pt) node[align=left,xshift = 0.35cm] {$\sigma_2$};
        \draw [thick] (1,1.5) -- (1,2.5);
        \filldraw (0,3.8) circle (2pt) node[align=left,yshift = 0.25cm] {$\rho_1$};
        \draw [thick] (1,2.5) -- (0,3.8);
        \filldraw (-1,3.8) circle (2pt) node[align=left,yshift = 0.25cm] {$\rho_2$};
        \draw [thick] (0,3.8) -- (-1,3.8);
        \draw [thick] (-1,3.8) -- (-2,2.50);
        \end{tikzpicture}
        \caption{Subcase (2+2+2+2) where the absolute values are pairwise unequal.}
        \label{f:4copies2222-1}
        \end{figure}

Due to their adjacent copies of $BP_{k-1}$ we know that $\pi_1$ begins with $\ul{q}$, $\pi_2$ begins with $\uline{t}$, and so on. In order for $\pi_2$ to start with $\ul{t}$ and $\rho_1$ to start with $\ul{s}$, which is $2$ flips away from $\pi_1$, it follows that $\pi_1$ must contain both $t$ and $\uline{s}$. This leads to two possibilities, either $\ul{s}$ precedes $t$ in $\pi_1$ or $t$ precedes $\ul{s}$ in $\pi_1$.

If $\pi_1=[\uline{q} X \uline{s} Y t Z p]$, then $\pi_2=[\uline{t} \overline{Y} s \overline{X} q Z p]$ and $\tau_1=[\uline{p} \overline{Z} \uline{q} X \uline{s} Y t]$. However, $\tau_2$ must have $\uline{s}$ in its first position, which is not possible by one flip from $\tau_1$.

If $\pi_1=[\ul{q} X t Y \ul{s} Z p]$, then $\pi_2=\pi_1 r_{|X|+2}
=[\uline{t} \overline{X} q Y \ul{s} Z p]$, $\tau_1=\pi_2 r_{k}
=[\uline{p} \overline{Z} s \overline{Y} \uline{q} X t]$, $\tau_2=\tau_1 r_{|Z|+2}
=[\uline{s} Z p \overline{Y} \uline{q} X t]$, $\sigma_1=\tau_2 r_{k}
=[\uline{t} \overline{X} q Y \uline{p} \overline{Z} s]$, $\sigma_2=\sigma_1 r_{|X|+2}
=[\uline{q} X t Y \uline{p} \overline{Z} s]$, $\rho_1=\sigma_2 r_{k}
=[\uline{s} Z p \overline{Y} \uline{t} \overline{X} q]$, $\rho_2=\rho_1 r_{|Z|+2}
=[\uline{p} \overline{Z} s \overline{Y} \uline{t} \overline{X} q]$, and $\pi_1=\rho_2 r_{k}$. Taking $|X|=i-2 \geq 0$ and $|Z|=j-2 \geq 0$ we get $|Y|=k-i-j \geq 0$ and a cycle corresponding to (\ref{8-2}).
    
\noindent \textbf{Subcase the absolute value of one pair among $p,q,s,t$ is the same.} 
Suppose, without loss of generality, that $|p|=|s|$, or more precisely $s=\ul{p}$. It follows from Lemma~\ref{lem:31} that the corresponding copies are non-adjacent (see Figure~\ref{f:4copies2222-2}).

     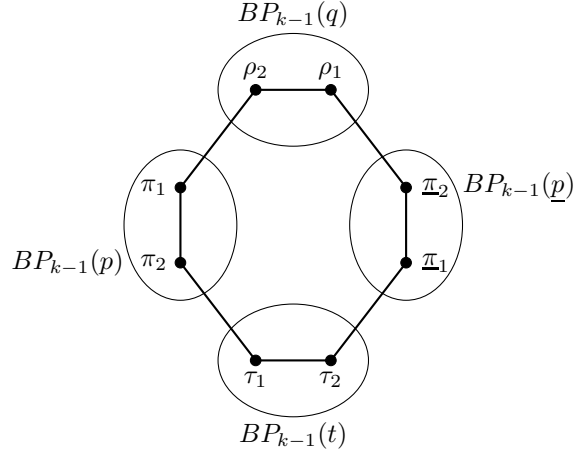
\begin{figure}
        \centering
        \begin{tikzpicture}[trim left=-4.cm]
        \draw (-2,2) ellipse (0.75cm and 1cm);
        \draw node at (-3.5,1.5) {$BP_{k-1}(p)$};
        \draw (1,2) ellipse (0.75cm and 1cm);
        \draw node at (-0.5,4.8) {$BP_{k-1}(q)$};
        \draw (-0.5,3.8) ellipse (1cm and 0.75cm);
        \draw node at (2.5,2.5) {$BP_{k-1}(\uline{p})$};
        \draw (-0.5,0.2) ellipse (1cm and 0.75cm);
        \draw node at (-0.5,-0.8) {$BP_{k-1}(t)$};
        \filldraw (-2,2.50) circle (2pt) node[align=left,xshift = -0.35cm] {$\pi_1$};
        \filldraw (-2,1.50) circle (2pt) node[align=left,xshift = -0.35cm] {$\pi_2$};
        \draw [thick] (-2,2.50) -- (-2,1.50);
        \filldraw (-1,0.2) circle (2pt) node[align=left,yshift = -0.25cm] {$\tau_1$};
        \draw [thick] (-2,1.50) -- (-1,0.20);
        \filldraw (0,0.2) circle (2pt) node[align=left,yshift = -0.25cm] {$\tau_2$};
        \draw [thick] (-1,0.20) -- (0,0.2);
        \filldraw (1,1.5) circle (2pt) node[align=left,xshift = 0.4cm] {$\ul{\pi}_1$};
        \draw [thick] (0,0.2) -- (1,1.50);
        \filldraw (1,2.5) circle (2pt) node[align=left,xshift = 0.4cm] {$\ul{\pi}_2$};
        \draw [thick] (1,1.5) -- (1,2.5);
        \filldraw (0,3.8) circle (2pt) node[align=left,yshift = 0.25cm] {$\rho_1$};
        \draw [thick] (1,2.5) -- (0,3.8);
        \filldraw (-1,3.8) circle (2pt) node[align=left,yshift = 0.25cm] {$\rho_2$};
        \draw [thick] (0,3.8) -- (-1,3.8);
        \draw [thick] (-1,3.8) -- (-2,2.50);
        \end{tikzpicture}
        \caption{Subcase (2+2+2+2) where one pair has equal absolute value.}
        \label{f:4copies2222-2}
        \end{figure}

Due to their adjacent copies of $BP_{k-1}$, we know that $\pi_1$ begins with $\ul{q}$, $\pi_2$ begins with $\ul{t}$, and so on. In order for $\pi_2$ to begin with $\ul{t}$, then $\pi_1$ must contain $t$. Thus, $\pi_1=[\uline{q} X t Y p]$. This gives $\pi_2=\pi_1 r_{|X|+2}
=[\uline{t} \overline{X} q Y p]$, $\tau_1=\pi_2 r_{k}
=[\uline{p} \overline{Y} \uline{q} X t]$, $\tau_2=\tau_1 r_1
=[p \overline{Y} \uline{q} X t]$, $\ul{\pi}_1=\tau_2 r_{k}
= [\uline{t} \overline{X} q Y \uline{p}]$, $\ul{\pi}_2=\ul{\pi}_1 r_{|X|+2}
=[\uline{q} X t Y \uline{p}]$, $\rho_1=\ul{\pi}_2 r_{k}
=[p\overline{Y} \uline{t} \overline{X} q]$, $\rho_2=\rho_1 r_1
=[\uline{p} \overline{Y} \uline{t} \overline{X} q]$ and $\pi_1=\rho_2 r_{k}$. Taking $|X|=i-2 \geq 0$, we get $|Y|=k-i-1 \geq 0$ and a cycle corresponding to (\ref{8-3}).

\noindent \textbf{Subcase the absolute values of two pairs among $p,q,s,t$ are the same.}
We say that $s=\uline{p}$ and $t=\uline{q}$ (See Figure~\ref{f:4copies2222-3}).

    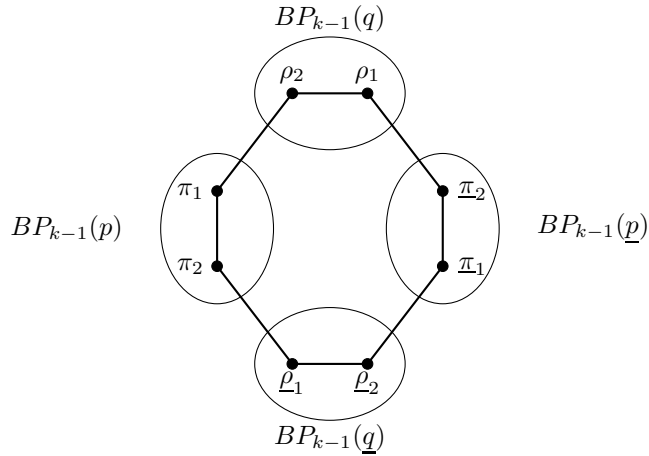
\begin{figure}[t]
    \centering
    \begin{tikzpicture}
        \draw (2,2) ellipse (0.75cm and 1cm);
        \draw node at (0,2) {$BP_{k-1}(p)$};
        \draw (5,2) ellipse (0.75cm and 1cm);
        \draw node at (3.5,4.8) {$BP_{k-1}(q)$};
        \draw (3.5,3.8) ellipse (1cm and 0.75cm);
        \draw node at (7,2) {$BP_{k-1}(\uline{p})$};
        \draw (3.5,0.2) ellipse (1cm and 0.75cm);
        \draw node at (3.5,-0.8) {$BP_{k-1}(\uline{q})$};
        \filldraw (2,2.50) circle (2pt) node[align=left,xshift = -0.35cm] {$\pi_1$};
        \filldraw (2,1.50) circle (2pt) node[align=left,xshift = -0.35cm] {$\pi_2$};
        \draw [thick] (2,2.50) -- (2,1.50);
        \filldraw (3,0.2) circle (2pt) node[align=left,yshift = -0.25cm] {$\ul{\rho}_1$};
        \draw [thick] (2,1.50) -- (3,0.20);
        \filldraw (4,0.2) circle (2pt) node[align=left,yshift = -0.25cm] {$\ul{\rho}_2$};
        \draw [thick] (3,0.20) -- (4,0.2);
        \filldraw (5,1.5) circle (2pt) node[align=left,xshift = 0.40cm] {$\ul{\pi}_1$};
        \draw [thick] (4,0.2) -- (5,1.50);
        \filldraw (5,2.5) circle (2pt) node[align=left,xshift = 0.40cm] {$\ul{\pi}_2$};
        \draw [thick] (5,1.5) -- (5,2.5);
        \filldraw (4,3.8) circle (2pt) node[align=left,yshift = 0.25cm] {$\rho_1$};
        \draw [thick] (5,2.5) -- (4,3.8);
        \filldraw (3,3.8) circle (2pt) node[align=left,yshift = 0.25cm] {$\rho_2$};
        \draw [thick] (4,3.8) -- (3,3.8);
        \draw [thick] (3,3.8) -- (2,2.50);
    \end{tikzpicture}
    \caption{Subcase (2+2+2+2) where two pairs have equal absolute values.}
    \label{f:4copies2222-3}
    \end{figure}
        
Let $\pi_1=[\uline{q} X p]$. Then, we get $\pi_2=\pi_1
r_1
=[q X p]$, $\ul{\rho}_1=\pi_2
r_{k}
=[\uline{p} \overline{X} \uline{q}]$, $\ul{\rho}_2=\ul{\rho}_1
r_1
=[p \overline{X} \uline{q}]$, $\ul{\pi}_1=\ul{\rho}_2
r_{k}
=[q X \uline{p}]$, $\ul{\pi}_2=\ul{\pi}_1
r_1
=[\uline{q} X \uline{p}]$, $\rho_1=\ul{\pi}_2
r_{k}
=[p \overline{X} q]$, $\rho_2=\rho_1
r_1
=[\uline{p} \overline{X} q]$, and $\pi_1=\rho_2
r_{k}$. This gives a cycle corresponding to (\ref{8-4}).

Having considered all the possible partitions corresponding to $8$-cycles, we conclude that our classification is now complete. \hfill$\square$

\section{Conclusions}\label{s:discussion}

The first main result of this paper is that the burnt pancake graph $BP_n$ for the group of signed permutations $B_n$, with $n\geq2$, is Hamiltonian and weakly pancyclic. We found this result to be a bit surprising since $BP_n$ is fairly sparse. To make this observation more formal, recall that if $G=(V,E)$ is an undirected simple graph, its \textit{edge density} $\delta(G)$ is defined as
\[
\delta(G):=\frac{|E|}{\binom{|V|}{2}}.
\] That is, $\delta(G)$ is the ratio of its edges to all the potential edges (see Diestel~\cite[Chapter 7]{D12}). The pancake graph $P_n$ of $\sn$ has $n!$ vertices and $\frac{n!(n-1)}{2}$ edges. Furthermore, $P_n$ is ($n-1$)-regular, whereas $BP_{n}$ is $n$-regular. Direct computation yields 
\[
\delta(P_n)=\frac{n-1}{n!-1}
\] and

\[
\delta(BP_n)=\frac{n}{2^nn!-1}.
\] In particular, one notices that $BP_n$ is significantly sparser than $P_n$, yet they are both Hamiltonian and weakly pancyclic. It is also worth note that the number of edges in $BP_n$ is significantly less than the bound on the number of edges compared to the number vertices in Bondy~\cite{Bondy71} of $|E|>|V|^2/4$ that would guarantee the graph is pancyclic. In terms of using $BP_n$ as an interconnection network, $BP_n$ has $n!$ times more vertices than the hypercube of the same degree (thus allowing for many more processors in the network), while preserving a linear diameter. Since $BP_n$ is Hamiltonian and weakly pancyclic, it offers another feature that is often desired in parallel processing, embedding of many different types of graphs. We have shown that the simplest types of graphs, cycles and paths, of length $8$ to $2^nn!$ can be embedded in the $BP_n$, for $n \geq 2$.

Our second main result is a classification of all of the $8$-cycles in $BP_n$, for $n\geq2$, which are the smallest cycles that one can embed in a burnt pancake graph. These classifications indicate how to embed an $8$-cycle within the network starting at any vertex. Through the proof, we can choose different canonical forms, depending on how many nodes are desired from copies of smaller burnt pancake graphs embedded in the greater network. We have already employed a similar method, using the recursive structure of $BP_n$, to classify the 9 and 10-cycles, but have not presented it here. We have used this description to obtain a formula for the number of signed permutations that require four flips to be sorted. At the moment, it seems that the number of stacks requiring $k$
flips to be sorted in the burnt pancake problem is given by an integer-valued polynomial of degree $k+1$.

\section*{Acknowledgments} 
The authors wish to thank the anonymous reviewers whose suggestions significantly improved this article.

\appendix

\section{Cycles of lengths from 8 to 48 in $BP_3$}\label{a:BP3}

We simply list the indices of the generators and write $i$ instead of $r_i$, with $1\leq i\leq 3$.

\begin{footnotesize}

\textbf{Length 8}
 \seqsplit{1 2 1 2 1 2 1 2}

\textbf{Length 9}
\seqsplit{1 2 1 3 1 2 3 2 3}

\textbf{Length 10}
\seqsplit{ 1 2 1 2 1 3 2 1 2 3}

\textbf{Length 11}
\seqsplit{1 2 1 2 3 1 2 3 2 3 2}

\textbf{Length 12}
\seqsplit{1 2 1 2 1 3 1 2 1 2 1 3}

\textbf{Length 13}
\seqsplit{1 2 1 2 1 2 3 1 2 1 3 2 3}

\textbf{Length 14}
\seqsplit{1 2 1 2 1 2 1 3 1 3 2 3 1 3}

\textbf{Length 15}
\seqsplit{1 2 1 2 1 2 1 3 1 2 1 3 1 2 3}

\textbf{Length 16}
\seqsplit{1 2 1 2 1 2 3 1 2 3 1 3 2 1 3 2}

\textbf{Length 17}
\seqsplit{1 2 1 2 1 2 1 3 1 2 3 1 3 1 3 2 3}

\textbf{Length 18}
\seqsplit{1 2 1 2 1 2 1 3 1 2 1 2 3 2 1 3 1 3}

\textbf{Length 19}
\seqsplit{1 2 1 2 1 2 1 3 1 2 1 2 1 2 3 1 2 1 3}

\textbf{Length 20}
\seqsplit{1 2 1 2 1 2 1 3 1 2 3 1 3 2 3 1 2 3 1 3}

\textbf{Length 21}
\seqsplit{1 2 1 2 1 2 1 3 1 2 1 2 1 2 3 2 1 2 1 2 3}

\textbf{Length 22}
\seqsplit{ 1 2 1 2 1 2 1 3 1 2 1 2 3 1 2 1 2 1 2 3 1 3}

\textbf{Length 23}
\seqsplit{ 1 2 1 2 1 2 1 3 1 2 1 2 1 2 1 3 1 2 3 1 3 1 3}

\textbf{Length 24}
\seqsplit{ 1 2 1 2 1 2 1 3 1 2 1 2 1 2 1 3 2 3 2 3 2 3 2 3}

\textbf{Length 25}
\seqsplit{ 1 2 1 2 1 2 1 3 1 2 1 2 1 2 1 3 1 2 1 2 3 1 3 2 3}

\textbf{Length 26}
\seqsplit{ 1 2 1 2 1 2 1 3 1 2 1 2 1 2 1 3 1 2 3 2 1 3 2 3 2 3}

\textbf{Length 27}
\seqsplit{ 1 2 1 2 1 2 1 3 1 2 1 2 1 2 1 3 1 3 1 3 1 2 1 2 1 2 3}

\textbf{Length 28}
\seqsplit{ 1 2 1 2 1 2 1 3 1 2 1 2 1 2 1 3 1 2 1 2 3 2 3 1 2 3 1 3}

\textbf{Length 29}
\seqsplit{ 1 2 1 2 1 2 1 3 1 2 1 2 1 2 1 3 1 2 1 2 3 1 2 1 2 3 2 1 3}

\textbf{Length 30}
\seqsplit{ 1 2 1 2 1 2 1 3 1 2 1 2 1 2 1 3 1 2 1 2 1 2 1 3 2 3 2 3 2 3}

\textbf{Length 31}
\seqsplit{ 1 2 1 2 1 2 1 3 1 2 1 2 1 2 1 3 1 2 1 2 1 2 1 3 1 2 1 2 3 1 3}

\textbf{Length 32}
\seqsplit{ 1 2 1 2 1 2 1 3 1 2 1 2 1 2 1 3 1 2 1 2 1 2 1 3 1 2 3 2 1 3 2 3}

\textbf{Length 33}
\seqsplit{ 1 2 1 2 1 2 1 3 1 2 1 2 1 2 1 3 1 2 1 2 3 1 2 1 2 3 1 2 1 2 1 2 3}

\textbf{Length 34}
\seqsplit{ 1 2 1 2 1 2 1 3 1 2 1 2 1 2 1 3 1 2 1 2 1 2 1 3 1 2 3 1 2 1 3 2 1 3}

\textbf{Length 35}
\seqsplit{ 1 2 1 2 1 2 1 3 1 2 1 2 1 2 1 3 1 2 1 2 1 2 1 3 1 3 2 1 3 2 1 2 3 2 3}

\textbf{Length 36}
\seqsplit{ 1 2 1 2 1 2 1 3 1 2 1 2 1 2 1 3 1 2 1 2 1 2 1 3 1 2 1 2 1 2 1 3 2 3 2 3}

\textbf{Length 37}
\seqsplit{ 1 2 1 2 1 2 1 3 1 2 1 2 1 2 1 3 1 2 1 2 1 2 1 3 1 2 1 2 3 2 1 2 1 2 1 2 3}

\textbf{Length 38}
\seqsplit{ 1 2 1 2 1 2 1 3 1 2 1 2 1 2 1 3 1 2 1 2 1 2 1 3 1 2 1 2 1 2 1 3 1 2 3 2 1 3}

\textbf{Length 39}
\seqsplit{ 1 2 1 2 1 2 1 3 1 2 1 2 1 2 1 3 1 2 1 2 1 2 1 3 1 2 1 2 3 2 3 1 2 1 2 1 3 2 3}

\textbf{Length 40}
\seqsplit{ 1 2 1 2 1 2 1 3 1 2 1 2 1 2 1 3 1 2 1 2 1 2 1 3 1 2 3 2 1 2 1 2 3 1 2 1 2 1 2 3}

\textbf{Length 41}
\seqsplit{ 1 2 1 2 1 2 1 3 1 2 1 2 1 2 1 3 1 2 1 2 1 2 1 3 1 2 3 1 2 1 3 2 3 1 2 1 3 2 3 2 3}

\textbf{Length 42}
\seqsplit{ 1 2 1 2 1 2 1 3 1 2 1 2 1 2 1 3 1 2 1 2 1 2 1 3 1 2 1 2 1 2 1 3 1 2 1 2 1 2 1 3 2 3}

\textbf{Length 43}
\seqsplit{ 1 2 1 2 1 2 1 3 1 2 1 2 1 2 1 3 1 2 1 2 1 2 1 3 1 2 3 2 1 2 3 1 2 1 3 2 1 2 1 2 1 2 3}

\textbf{Length 44}
\seqsplit{ 1 2 1 2 1 2 1 3 1 2 1 2 1 2 1 3 1 2 1 2 1 2 1 3 1 3 1 2 1 3 1 2 1 2 1 2 1 3 1 2 1 3 1 3}

\textbf{Length 45}
\seqsplit{ 1 2 1 2 1 2 1 3 1 2 1 2 1 2 1 3 1 2 1 2 1 2 1 3 1 2 3 1 2 1 2 1 3 1 2 1 3 2 1 2 1 2 1 2 3}

\textbf{Length 46}
\seqsplit{ 1 2 1 2 1 2 1 3 1 2 1 2 1 2 1 3 1 2 1 2 1 2 1 3 1 2 1 2 1 2 1 3 2 1 2 1 2 1 3 1 2 1 2 1 2 3}

\textbf{Length 47}
\seqsplit{ 1 2 1 2 1 2 1 3 1 2 1 2 3 2 1 2 1 2 3 1 2 1 3 1 2 1 2 1 2 1 3 2 3 1 2 1 2 1 2 1 3 1 3 1 2 1 3}

\textbf{Length 48}
\seqsplit{ 1 2 1 2 1 2 1 3 1 2 1 2 1 2 1 3 1 2 1 2 1 2 1 3 1 2 1 2 1 2 1 3 1 2 1 2 1 2 1 3 1 2 1 2 1 2 1 3}

\end{footnotesize}

\section{Cycles of lengths from 8 to 384 in $BP_4$}\label{a:BP4}

We simply list the indices of the generators and write $i$ instead of $r_i$, with $1\leq i\leq 4$.

\begin{footnotesize}

\textbf{Length 8}
\seqsplit{12121212}

\textbf{Length 9}
\seqsplit{121312323}

\textbf{Length 10}
\seqsplit{1212132123}

\textbf{Length 11}
\seqsplit{12123123232}

\textbf{Length 12}
\seqsplit{121213121213}

\textbf{Length 13}
\seqsplit{1212123121323}

\textbf{Length 14}
\seqsplit{12121213132313}

\textbf{Length 15}
\seqsplit{121212131213123}

\textbf{Length 16}
\seqsplit{1212121313242124}

\textbf{Length 17}
\seqsplit{12121213121432134}

\textbf{Length 18}
\seqsplit{121212131212321313}

\textbf{Length 19}
\seqsplit{1212121312121231213}

\textbf{Length 20}
\seqsplit{12121213121212412324}

\textbf{Length 21}
\seqsplit{121212131212123212123}

\textbf{Length 22}
\seqsplit{1212121312121234234343}

\textbf{Length 23}
\seqsplit{12121213121212131231313}

\textbf{Length 24}
\seqsplit{121212131212121323232323}

\textbf{Length 25}
\seqsplit{1212121312121213121231323}

\textbf{Length 26}
\seqsplit{12121213121212131212412434}

\textbf{Length 27}
\seqsplit{121212131212121312123143234}

\textbf{Length 28}
\seqsplit{1212121312121213121214124134}

\textbf{Length 29}
\seqsplit{12121213121212131212142124123}

\textbf{Length 30}
\seqsplit{121212131212121312121213232323}

\textbf{Length 31}
\seqsplit{1212121312121213121212131212313}

\textbf{Length 32}
\seqsplit{12121213121212131212121312321323}

\textbf{Length 33}
\seqsplit{121212131212121312121213121242124}

\textbf{Length 34}
\seqsplit{1212121312121213121212131214124343}

\textbf{Length 35}
\seqsplit{12121213121212131212121312124121214}

\textbf{Length 36}
\seqsplit{121212131212121312121213121212132323}

\textbf{Length 37}
\seqsplit{1212121312121213121212131212142141424}

\textbf{Length 38}
\seqsplit{12121213121212131212121312121213123213}

\textbf{Length 39}
\seqsplit{121212131212121312121213121212413423234}

\textbf{Length 40}
\seqsplit{1212121312121213121212131212121312431234}

\textbf{Length 41}
\seqsplit{12121213121212131212121312121213124121324}

\textbf{Length 42}
\seqsplit{121212131212121312121213121212131212121323}

\textbf{Length 43}
\seqsplit{1212121312121213121212131212121312124234234}

\textbf{Length 44}
\seqsplit{12121213121212131212121312121213121212143234}

\textbf{Length 45}
\seqsplit{121212131212121312121213121212131212124134243}

\textbf{Length 46}
\seqsplit{1212121312121213121212131212121312121214132314}

\textbf{Length 47}
\seqsplit{12121213121212131212121312121213121212134121434}

\textbf{Length 48}
\seqsplit{121212131212121312121213121212131212121312121213}

\textbf{Length 49}
\seqsplit{1212121312121213121212131212121312121213141214134}

\textbf{Length 50}
\seqsplit{12121213121212131212121312121213121212131214313413}

\textbf{Length 51}
\seqsplit{121212131212121312121213121212131212121312124234343}

\textbf{Length 52}
\seqsplit{1212121312121213121212131212121312121213121214324324}

\textbf{Length 53}
\seqsplit{12121213121212131212121312121213121212131212124132424}

\textbf{Length 54}
\seqsplit{121212131212121312121213121212131212121312121214143414}

\textbf{Length 55}
\seqsplit{1212121312121213121212131212121312121213121212141214234}

\textbf{Length 56}
\seqsplit{12121213121212131212121312121213121212131212124124142143}

\textbf{Length 57}
\seqsplit{121212131212121312121213121212131212121312121214124241434}

\textbf{Length 58} \seqsplit{1212121312121213121212131212121312121213121212141212431414}

\textbf{Length 59} \seqsplit{12121213121212131212121312121213121212131212121412121242314}

\textbf{Length 60} \seqsplit{121212131212121312121213121212131212121312121214121231413434}

\textbf{Length 61} \seqsplit{1212121312121213121212131212121312121213121212141212124313234}

\textbf{Length 62} \seqsplit{12121213121212131212121312121213121212131212121412121231234124}

\textbf{Length 63} \seqsplit{121212131212121312121213121212131212121312121214121212131424324}

\textbf{Length 64} \seqsplit{1212121312121213121212131212121312121213121212141212121312314124}

\textbf{Length 65} \seqsplit{12121213121212131212121312121213121212131212121412121213124242324}

\textbf{Length 66} \seqsplit{121212131212121312121213121212131212121312121214121212131213141214}

\textbf{Length 67} \seqsplit{1212121312121213121212131212121312121213121212141212121312123142134}

\textbf{Length 68} \seqsplit{12121213121212131212121312121213121212131212121412121213121241321434}

\textbf{Length 69} \seqsplit{121212131212121312121213121212131212121312121214121212131212142413214}

\textbf{Length 70} \seqsplit{1212121312121213121212131212121312121213121212141212121312121412423234}

\textbf{Length 71} \seqsplit{12121213121212131212121312121213121212131212121412121213121212132413434}

\textbf{Length 72} \seqsplit{121212131212121312121213121212131212121312121214121212131212121312124124}

\textbf{Length 73} \seqsplit{1212121312121213121212131212121312121213121212141212121312121213124131434}

\textbf{Length 74} \seqsplit{12121213121212131212121312121213121212131212121412121213121212131214142424}

\textbf{Length 75} \seqsplit{121212131212121312121213121212131212121312121214121212131212121312123432414}

\textbf{Length 76} \seqsplit{1212121312121213121212131212121312121213121212141212121312121213121214141424}

\textbf{Length 77} \seqsplit{12121213121212131212121312121213121212131212121412121213121212131212121413434}

\textbf{Length 78} \seqsplit{121212131212121312121213121212131212121312121214121212131212121312121213242134}

\textbf{Length 79} \seqsplit{1212121312121213121212131212121312121213121212141212121312121213121212141234324}

\textbf{Length 80} \seqsplit{12121213121212131212121312121213121212131212121412121213121212131212121312421314}

\textbf{Length 81} \seqsplit{121212131212121312121213121212131212121312121214121212131212121312121213124132324}

\textbf{Length 82} \seqsplit{1212121312121213121212131212121312121213121212141212121312121213121212131232432414}

\textbf{Length 83} \seqsplit{12121213121212131212121312121213121212131212121412121213121212131212121312321341214}

\textbf{Length 84} \seqsplit{121212131212121312121213121212131212121312121214121212131212121312121213121212142134}

\textbf{Length 85} \seqsplit{1212121312121213121212131212121312121213121212141212121312121213121212131212142123414}

\textbf{Length 86} \seqsplit{12121213121212131212121312121213121212131212121412121213121212131212121312121242141434}

\textbf{Length 87} \seqsplit{121212131212121312121213121212131212121312121214121212131212121312121213121212132341214}

\textbf{Length 88} \seqsplit{1212121312121213121212131212121312121213121212141212121312121213121212131212121321432414}

\textbf{Length 89} \seqsplit{12121213121212131212121312121213121212131212121412121213121212131212121312121213141324134}

\textbf{Length 90} \seqsplit{121212131212121312121213121212131212121312121214121212131212121312121213121212131214312414}

\textbf{Length 91} \seqsplit{1212121312121213121212131212121312121213121212141212121312121213121212131212121312321241214}

\textbf{Length 92} \seqsplit{12121213121212131212121312121213121212131212121412121213121212131212121312121213121212432414}

\textbf{Length 93} \seqsplit{121212131212121312121213121212131212121312121214121212131212121312121213121212131212121341214}

\textbf{Length 94} \seqsplit{1212121312121213121212131212121312121213121212141212121312121213121212131212121312121242142324}

\textbf{Length 95} \seqsplit{12121213121212131212121312121213121212131212121412121213121212131212121312121213121212134212124}

\textbf{Length 96} \seqsplit{121212131212121312121213121212131212121312121214121212131212121312121213121212131212121343123234}

\textbf{Length 97} \seqsplit{1212121312121213121212131212121312121213121212141212121312121213121212131212121312121213141241414}

\textbf{Length 98} \seqsplit{12121213121212131212121312121213121212131212121412121213121212131212121312121213121212131214134314}

\textbf{Length 99} \seqsplit{121212131212121312121213121212131212121312121214121212131212121312121213121212131212121312412421414}

\textbf{Length 100} \seqsplit{1212121312121213121212131212121312121213121212141212121312121213121212131212121312121213121214323234}

\textbf{Length 101} \seqsplit{12121213121212131212121312121213121212131212121412121213121212131212121312121213121212131212141213134}

\textbf{Length 102} \seqsplit{121212131212121312121213121212131212121312121214121212131212121312121213121212131212121312121241314314}

\textbf{Length 103} \seqsplit{1212121312121213121212131212121312121213121212141212121312121213121212131212121312121213121212141242414}

\textbf{Length 104} \seqsplit{12121213121212131212121312121213121212131212121412121213121212131212121312121213121212131212121413141314}

\textbf{Length 105} \seqsplit{1212121312121213121212131212121312121213121212141212121312121213121212131212121312121213121212141212424
24}

\textbf{Length 106} \seqsplit{1212121312121213121212131212121312121213121212141212121312121213121212131212121312121213121212141212314
134}

\textbf{Length 107} \seqsplit{1212121312121213121212131212121312121213121212141212121312121213121212131212121312121213121212141231342
1414}

\textbf{Length 108} \seqsplit{1212121312121213121212131212121312121213121212141212121312121213121212131212121312121213121212141212341
41434}

\textbf{Length 109} \seqsplit{1212121312121213121212131212121312121213121212141212121312121213121212131212121312121213121212141212313
421424}

\textbf{Length 110} \seqsplit{1212121312121213121212131212121312121213121212141212121312121213121212131212121312121213121212141212121
4343424}

\textbf{Length 111} \seqsplit{1212121312121213121212131212121312121213121212141212121312121213121212131212121312121213121212141212121
31423234}

\textbf{Length 112} \seqsplit{1212121312121213121212131212121312121213121212141212121312121213121212131212121312121213121212141212121
314131214}

\textbf{Length 113} \seqsplit{1212121312121213121212131212121312121213121212141212121312121213121212131212121312121213121212141212121
3123214134}

\textbf{Length 114} \seqsplit{1212121312121213121212131212121312121213121212141212121312121213121212131212121312121213121212141212121
31212413214}

\textbf{Length 115} \seqsplit{1212121312121213121212131212121312121213121212141212121312121213121212131212121312121213121212141212121
312123421414}

\textbf{Length 116} \seqsplit{1212121312121213121212131212121312121213121212141212121312121213121212131212121312121213121212141212121
3121231234134}

\textbf{Length 117} \seqsplit{1212121312121213121212131212121312121213121212141212121312121213121212131212121312121213121212141212121
31212121324134}

\textbf{Length 118} \seqsplit{1212121312121213121212131212121312121213121212141212121312121213121212131212121312121213121212141212121
312121214341234}

\textbf{Length 119} \seqsplit{1212121312121213121212131212121312121213121212141212121312121213121212131212121312121213121212141212121
3121212131241314}

\textbf{Length 120} \seqsplit{1212121312121213121212131212121312121213121212141212121312121213121212131212121312121213121212141212121
31212121323421424}

\textbf{Length 121} \seqsplit{1212121312121213121212131212121312121213121212141212121312121213121212131212121312121213121212141212121
312121213124313134}

\textbf{Length 122} \seqsplit{1212121312121213121212131212121312121213121212141212121312121213121212131212121312121213121212141212121
3121212131232421414}

\textbf{Length 123} \seqsplit{1212121312121213121212131212121312121213121212141212121312121213121212131212121312121213121212141212121
31212121312121214134}

\textbf{Length 124} \seqsplit{1212121312121213121212131212121312121213121212141212121312121213121212131212121312121213121212141212121
312121213121214123414}

\textbf{Length 125} \seqsplit{1212121312121213121212131212121312121213121212141212121312121213121212131212121312121213121212141212121
3121212131212124141434}

\textbf{Length 126} \seqsplit{1212121312121213121212131212121312121213121212141212121312121213121212131212121312121213121212141212121
31212121312121213421424}

\textbf{Length 127} \seqsplit{1212121312121213121212131212121312121213121212141212121312121213121212131212121312121213121212141212121
312121213121212132341324}

\textbf{Length 128} \seqsplit{1212121312121213121212131212121312121213121212141212121312121213121212131212121312121213121212141212121
3121212131212121321421414}

\textbf{Length 129} \seqsplit{1212121312121213121212131212121312121213121212141212121312121213121212131212121312121213121212141212121
31212121312121213123413124}

\textbf{Length 130} \seqsplit{1212121312121213121212131212121312121213121212141212121312121213121212131212121312121213121212141212121
312121213121212131234213234}

\textbf{Length 131} \seqsplit{1212121312121213121212131212121312121213121212141212121312121213121212131212121312121213121212141212121
3121212131212121312143431424}

\textbf{Length 132} \seqsplit{1212121312121213121212131212121312121213121212141212121312121213121212131212121312121213121212141212121
31212121312121213121212421414}

\textbf{Length 133} \seqsplit{1212121312121213121212131212121312121213121212141212121312121213121212131212121312121213121212141212121
312121213121212131212121341324}

\textbf{Length 134} \seqsplit{1212121312121213121212131212121312121213121212141212121312121213121212131212121312121213121212141212121
3121212131212121312121243242324}

\textbf{Length 135} \seqsplit{1212121312121213121212131212121312121213121212141212121312121213121212131212121312121213121212141212121
31212121312121213121212134323134}

\textbf{Length 136} \seqsplit{1212121312121213121212131212121312121213121212141212121312121213121212131212121312121213121212141212121
312121213121212131212121342131234}

\textbf{Length 137} \seqsplit{1212121312121213121212131212121312121213121212141212121312121213121212131212121312121213121212141212121
3121212131212121312121213141341424}

\textbf{Length 138} \seqsplit{1212121312121213121212131212121312121213121212141212121312121213121212131212121312121213121212141212121
31212121312121213121212131214123424}

\textbf{Length 139} \seqsplit{1212121312121213121212131212121312121213121212141212121312121213121212131212121312121213121212141212121
312121213121212131212121312121413124}

\textbf{Length 140} \seqsplit{1212121312121213121212131212121312121213121212141212121312121213121212131212121312121213121212141212121
3121212131212121312121213121214213234}

\textbf{Length 141} \seqsplit{1212121312121213121212131212121312121213121212141212121312121213121212131212121312121213121212141212121
31212121312121213121212131212143123134}

\textbf{Length 142} \seqsplit{1212121312121213121212131212121312121213121212141212121312121213121212131212121312121213121212141212121
312121213121212131212121312121241231424}

\textbf{Length 143} \seqsplit{1212121312121213121212131212121312121213121212141212121312121213121212131212121312121213121212141212121
3121212131212121312121213121212414124324}

\textbf{Length 144} \seqsplit{1212121312121213121212131212121312121213121212141212121312121213121212131212121312121213121212141212121
31212121312121213121212131212121314342324}

\textbf{Length 145} \seqsplit{1212121312121213121212131212121312121213121212141212121312121213121212131212121312121213121212141212121
312121213121212131212121312121213141213414}

\textbf{Length 146} \seqsplit{1212121312121213121212131212121312121213121212141212121312121213121212131212121312121213121212141212121
3121212131212121312121213121212131214121234}

\textbf{Length 147} \seqsplit{1212121312121213121212131212121312121213121212141212121312121213121212131212121312121213121212141212121
31212121312121213121212131212121312142312324}

\textbf{Length 148} \seqsplit{1212121312121213121212131212121312121213121212141212121312121213121212131212121312121213121212141212121
312121213121212131212121312121213121243142324}

\textbf{Length 149} \seqsplit{1212121312121213121212131212121312121213121212141212121312121213121212131212121312121213121212141212121
3121212131212121312121213121212131212124232324}

\textbf{Length 150} \seqsplit{1212121312121213121212131212121312121213121212141212121312121213121212131212121312121213121212141212121
31212121312121213121212131212121312121241212314}

\textbf{Length 151} \seqsplit{1212121312121213121212131212121312121213121212141212121312121213121212131212121312121213121212141212121
312121213121212131212121312121213121212412321324}

\textbf{Length 152} \seqsplit{1212121312121213121212131212121312121213121212141212121312121213121212131212121312121213121212141212121
3121212131212121312121213121212131212124121313234}

\textbf{Length 153} \seqsplit{1212121312121213121212131212121312121213121212141212121312121213121212131212121312121213121212141212121
31212121312121213121212131212121312121214142324324}

\textbf{Length 154} \seqsplit{1212121312121213121212131212121312121213121212141212121312121213121212131212121312121213121212141212121
312121213121212131212121312121213121212141212341414}

\textbf{Length 155} \seqsplit{1212121312121213121212131212121312121213121212141212121312121213121212131212121312121213121212141212121
3121212131212121312121213121212131212121412134342434}

\textbf{Length 156} \seqsplit{1212121312121213121212131212121312121213121212141212121312121213121212131212121312121213121212141212121
31212121312121213121212131212121312121214121231241424}

\textbf{Length 157} \seqsplit{1212121312121213121212131212121312121213121212141212121312121213121212131212121312121213121212141212121
312121213121212131212121312121213121212141212343242434}

\textbf{Length 158} \seqsplit{1212121312121213121212131212121312121213121212141212121312121213121212131212121312121213121212141212121
3121212131212121312121213121212131212121412121234243424}

\textbf{Length 159} \seqsplit{1212121312121213121212131212121312121213121212141212121312121213121212131212121312121213121212141212121
31212121312121213121212131212121312121214121212132341424}

\textbf{Length 160} \seqsplit{1212121312121213121212131212121312121213121212141212121312121213121212131212121312121213121212141212121
312121213121212131212121312121213121212141212121342143424}

\textbf{Length 161} \seqsplit{1212121312121213121212131212121312121213121212141212121312121213121212131212121312121213121212141212121
3121212131212121312121213121212131212121412121213123241414}

\textbf{Length 162} \seqsplit{1212121312121213121212131212121312121213121212141212121312121213121212131212121312121213121212141212121
31212121312121213121212131212121312121214121212131242132434}

\textbf{Length 163} \seqsplit{1212121312121213121212131212121312121213121212141212121312121213121212131212121312121213121212141212121
312121213121212131212121312121213121212141212121312321241424}

\textbf{Length 164} \seqsplit{1212121312121213121212131212121312121213121212141212121312121213121212131212121312121213121212141212121
3121212131212121312121213121212131212121412121213121231241434}

\textbf{Length 165} \seqsplit{1212121312121213121212131212121312121213121212141212121312121213121212131212121312121213121212141212121
31212121312121213121212131212121312121214121212131212121341424}

\textbf{Length 166} \seqsplit{1212121312121213121212131212121312121213121212141212121312121213121212131212121312121213121212141212121
312121213121212131212121312121213121212141212121312121234243434}

\textbf{Length 167} \seqsplit{1212121312121213121212131212121312121213121212141212121312121213121212131212121312121213121212141212121
3121212131212121312121213121212131212121412121213121212132141414}

\textbf{Length 168} \seqsplit{1212121312121213121212131212121312121213121212141212121312121213121212131212121312121213121212141212121
31212121312121213121212131212121312121214121212131212121342143434}

\textbf{Length 169} \seqsplit{1212121312121213121212131212121312121213121212141212121312121213121212131212121312121213121212141212121
312121213121212131212121312121213121212141212121312121213134143124}

\textbf{Length 170} \seqsplit{1212121312121213121212131212121312121213121212141212121312121213121212131212121312121213121212141212121
3121212131212121312121213121212131212121412121213121212131324243424}

\textbf{Length 171} \seqsplit{1212121312121213121212131212121312121213121212141212121312121213121212131212121312121213121212141212121
31212121312121213121212131212121312121214121212131212121312121241414}

\textbf{Length 172} \seqsplit{1212121312121213121212131212121312121213121212141212121312121213121212131212121312121213121212141212121
312121213121212131212121312121213121212141212121312121213121234232434}

\textbf{Length 173} \seqsplit{1212121312121213121212131212121312121213121212141212121312121213121212131212121312121213121212141212121
3121212131212121312121213121212131212121412121213121212131212121341434}

\textbf{Length 174} \seqsplit{1212121312121213121212131212121312121213121212141212121312121213121212131212121312121213121212141212121
31212121312121213121212131212121312121214121212131212121312121243242434}

\textbf{Length 175} \seqsplit{1212121312121213121212131212121312121213121212141212121312121213121212131212121312121213121212141212121
312121213121212131212121312121213121212141212121312121213121212423242324}

\textbf{Length 176} \seqsplit{1212121312121213121212131212121312121213121212141212121312121213121212131212121312121213121212141212121
3121212131212121312121213121212131212121412121213121212131212121342423414}

\textbf{Length 177} \seqsplit{1212121312121213121212131212121312121213121212141212121312121213121212131212121312121213121212141212121
31212121312121213121212131212121312121214121212131212121312121213214141234}

\textbf{Length 178} \seqsplit{1212121312121213121212131212121312121213121212141212121312121213121212131212121312121213121212141212121
312121213121212131212121312121213121212141212121312121213121212131324243434}

\textbf{Length 179} \seqsplit{1212121312121213121212131212121312121213121212141212121312121213121212131212121312121213121212141212121
3121212131212121312121213121212131212121412121213121212131212121312141412134}

\textbf{Length 180} \seqsplit{1212121312121213121212131212121312121213121212141212121312121213121212131212121312121213121212141212121
31212121312121213121212131212121312121214121212131212121312121213121412343434}

\textbf{Length 181} \seqsplit{1212121312121213121212131212121312121213121212141212121312121213121212131212121312121213121212141212121
312121213121212131212121312121213121212141212121312121213121212131212124141234}

\textbf{Length 182} \seqsplit{1212121312121213121212131212121312121213121212141212121312121213121212131212121312121213121212141212121
3121212131212121312121213121212131212121412121213121212131212121312121242413414}

\textbf{Length 183} \seqsplit{1212121312121213121212131212121312121213121212141212121312121213121212131212121312121213121212141212121
31212121312121213121212131212121312121214121212131212121312121213121212412321424}

\textbf{Length 184} \seqsplit{1212121312121213121212131212121312121213121212141212121312121213121212131212121312121213121212141212121
312121213121212131212121312121213121212141212121312121213121212131212121314342434}

\textbf{Length 185} \seqsplit{1212121312121213121212131212121312121213121212141212121312121213121212131212121312121213121212141212121
3121212131212121312121213121212131212121412121213121212131212121312121213141323414}

\textbf{Length 186} \seqsplit{1212121312121213121212131212121312121213121212141212121312121213121212131212121312121213121212141212121
31212121312121213121212131212121312121214121212131212121312121213121212131214143124}

\textbf{Length 187} \seqsplit{1212121312121213121212131212121312121213121212141212121312121213121212131212121312121213121212141212121
312121213121212131212121312121213121212141212121312121213121212131212121312124243424}

\textbf{Length 188} \seqsplit{1212121312121213121212131212121312121213121212141212121312121213121212131212121312121213121212141212121
3121212131212121312121213121212131212121412121213121212131212121312121213121241423124}

\textbf{Length 189} \seqsplit{1212121312121213121212131212121312121213121212141212121312121213121212131212121312121213121212141212121
31212121312121213121212131212121312121214121212131212121312121213121212131212124232434}

\textbf{Length 190} \seqsplit{1212121312121213121212131212121312121213121212141212121312121213121212131212121312121213121212141212121
312121213121212131212121312121213121212141212121312121213121212131212121312121241431214}

\textbf{Length 191} \seqsplit{1212121312121213121212131212121312121213121212141212121312121213121212131212121312121213121212141212121
3121212131212121312121213121212131212121412121213121212131212121312121213121212412321434}

\textbf{Length 192} \seqsplit{1212121312121213121212131212121312121213121212141212121312121213121212131212121312121213121212141212121
31212121312121213121212131212121312121214121212131212121312121213121212131212124121243214}

\textbf{Length 193} \seqsplit{1212121312121213121212131212121312121213121212141212121312121213121212131212121312121213121212141212121
312121213121212131212121312121213121212141212121312121213121212131212121312121213142434314}

\textbf{Length 194} \seqsplit{1212121312121213121212131212121312121213121212141212121312121213121212131212121312121213121212141212121
3121212131212121312121213121212131212121412121213121212131212121312121213121212131241212414}

\textbf{Length 195} \seqsplit{1212121312121213121212131212121312121213121212141212121312121213121212131212121312121213121212141212121
31212121312121213121212131212121312121214121212131212121312121213121212131212121312124243434}

\textbf{Length 196} \seqsplit{1212121312121213121212131212121312121213121212141212121312121213121212131212121312121213121212141212121
312121213121212131212121312121213121212141212121312121213121212131212121312121213121241212424}

\textbf{Length 197} \seqsplit{1212121312121213121212131212121312121213121212141212121312121213121212131212121312121213121212141212121
3121212131212121312121213121212131212121412121213121212131212121312121213121212131212412431434}

\textbf{Length 198} \seqsplit{1212121312121213121212131212121312121213121212141212121312121213121212131212121312121213121212141212121
31212121312121213121212131212121312121214121212131212121312121213121212131212121312124132123424}

\textbf{Length 199} \seqsplit{1212121312121213121212131212121312121213121212141212121312121213121212131212121312121213121212141212121
312121213121212131212121312121213121212141212121312121213121212131212121312121213121241213413414}

\textbf{Length 200} \seqsplit{1212121312121213121212131212121312121213121212141212121312121213121212131212121312121213121212141212121
3121212131212121312121213121212131212121412121213121212131212121312121213121212131212121434343434}

\textbf{Length 201} \seqsplit{1212121312121213121212131212121312121213121212141212121312121213121212131212121312121213121212141212121
31212121312121213121212131212121312121214121212131212121312121213121212131212121312121214121343424}

\textbf{Length 202} \seqsplit{1212121312121213121212131212121312121213121212141212121312121213121212131212121312121213121212141212121
312121213121212131212121312121213121212141212121312121213121212131212121312121213121212141212313414}

\textbf{Length 203} \seqsplit{1212121312121213121212131212121312121213121212141212121312121213121212131212121312121213121212141212121
3121212131212121312121213121212131212121412121213121212131212121312121213121212131212121412123432424}

\textbf{Length 204} \seqsplit{1212121312121213121212131212121312121213121212141212121312121213121212131212121312121213121212141212121
31212121312121213121212131212121312121214121212131212121312121213121212131212121312121214121213131414}

\textbf{Length 205} \seqsplit{1212121312121213121212131212121312121213121212141212121312121213121212131212121312121213121212141212121
312121213121212131212121312121213121212141212121312121213121212131212121312121213121212141212123431424}

\textbf{Length 206} \seqsplit{1212121312121213121212131212121312121213121212141212121312121213121212131212121312121213121212141212121
3121212131212121312121213121212131212121412121213121212131212121312121213121212131212121412121231242324}

\textbf{Length 207} \seqsplit{1212121312121213121212131212121312121213121212141212121312121213121212131212121312121213121212141212121
3121212131212121312121213121212131212121412121213121212131212121312121213121212131212121412121213232341
4}

\textbf{Length 208} \seqsplit{1212121312121213121212131212121312121213121212141212121312121213121212131212121312121213121212141212121
3121212131212121312121213121212131212121412121213121212131212121312121213121212131212121412121213124213
24}

\textbf{Length 209} \seqsplit{1212121312121213121212131212121312121213121212141212121312121213121212131212121312121213121212141212121
3121212131212121312121213121212131212121412121213121212131212121312121213121212131212121412121213123213
414}

\textbf{Length 210} \seqsplit{1212121312121213121212131212121312121213121212141212121312121213121212131212121312121213121212141212121
3121212131212121312121213121212131212121412121213121212131212121312121213121212131212121412121213121231
2414}

\textbf{Length 211} \seqsplit{1212121312121213121212131212121312121213121212141212121312121213121212131212121312121213121212141212121
3121212131212121312121213121212131212121412121213121212131212121312121213121212131212121412121213121321
31414}

\textbf{Length 212} \seqsplit{1212121312121213121212131212121312121213121212141212121312121213121212131212121312121213121212141212121
3121212131212121312121213121212131212121412121213121212131212121312121213121212131212121412121213121212
342434}

\textbf{Length 213} \seqsplit{1212121312121213121212131212121312121213121212141212121312121213121212131212121312121213121212141212121
3121212131212121312121213121212131212121412121213121212131212121312121213121212131212121412121213121212
1323414}

\textbf{Length 214} \seqsplit{1212121312121213121212131212121312121213121212141212121312121213121212131212121312121213121212141212121
3121212131212121312121213121212131212121412121213121212131212121312121213121212131212121412121213121212
13421434}

\textbf{Length 215} \seqsplit{1212121312121213121212131212121312121213121212141212121312121213121212131212121312121213121212141212121
3121212131212121312121213121212131212121412121213121212131212121312121213121212131212121412121213121212
313243424}

\textbf{Length 216} \seqsplit{1212121312121213121212131212121312121213121212141212121312121213121212131212121312121213121212141212121
3121212131212121312121213121212131212121412121213121212131212121312121213121212131212121412121213121212
1313423124}

\textbf{Length 217} \seqsplit{1212121312121213121212131212121312121213121212141212121312121213121212131212121312121213121212141212121
3121212131212121312121213121212131212121412121213121212131212121312121213121212131212121412121213121212
13123212414}

\textbf{Length 218} \seqsplit{1212121312121213121212131212121312121213121212141212121312121213121212131212121312121213121212141212121
3121212131212121312121213121212131212121412121213121212131212121312121213121212131212121412121213121212
131212342324}

\textbf{Length 219} \seqsplit{1212121312121213121212131212121312121213121212141212121312121213121212131212121312121213121212141212121
3121212131212121312121213121212131212121412121213121212131212121312121213121212131212121412121213121212
1312121213414}

\textbf{Length 220} \seqsplit{1212121312121213121212131212121312121213121212141212121312121213121212131212121312121213121212141212121
3121212131212121312121213121212131212121412121213121212131212121312121213121212131212121412121213121212
13121212432424}

\textbf{Length 221} \seqsplit{1212121312121213121212131212121312121213121212141212121312121213121212131212121312121213121212141212121
3121212131212121312121213121212131212121412121213121212131212121312121213121212131212121412121213121212
131212141324324}

\textbf{Length 222} \seqsplit{1212121312121213121212131212121312121213121212141212121312121213121212131212121312121213121212141212121
3121212131212121312121213121212131212121412121213121212131212121312121213121212131212121412121213121212
1312121241324214}

\textbf{Length 223} \seqsplit{1212121312121213121212131212121312121213121212141212121312121213121212131212121312121213121212141212121
3121212131212121312121213121212131212121412121213121212131212121312121213121212131212121412121213121212
13121212431431434}

\textbf{Length 224} \seqsplit{1212121312121213121212131212121312121213121212141212121312121213121212131212121312121213121212141212121
3121212131212121312121213121212131212121412121213121212131212121312121213121212131212121412121213121212
131212121313242434}

\textbf{Length 225} \seqsplit{1212121312121213121212131212121312121213121212141212121312121213121212131212121312121213121212141212121
3121212131212121312121213121212131212121412121213121212131212121312121213121212131212121412121213121212
1312121213123242324}

\textbf{Length 226} \seqsplit{1212121312121213121212131212121312121213121212141212121312121213121212131212121312121213121212141212121
3121212131212121312121213121212131212121412121213121212131212121312121213121212131212121412121213121212
13121212131214123434}

\textbf{Length 227} \seqsplit{1212121312121213121212131212121312121213121212141212121312121213121212131212121312121213121212141212121
3121212131212121312121213121212131212121412121213121212131212121312121213121212131212121412121213121212
131212121312121413134}

\textbf{Length 228} \seqsplit{1212121312121213121212131212121312121213121212141212121312121213121212131212121312121213121212141212121
3121212131212121312121213121212131212121412121213121212131212121312121213121212131212121412121213121212
1312121213121232142434}

\textbf{Length 229} \seqsplit{1212121312121213121212131212121312121213121212141212121312121213121212131212121312121213121212141212121
3121212131212121312121213121212131212121412121213121212131212121312121213121212131212121412121213121212
13121212131212141231324}

\textbf{Length 230} \seqsplit{1212121312121213121212131212121312121213121212141212121312121213121212131212121312121213121212141212121
3121212131212121312121213121212131212121412121213121212131212121312121213121212131212121412121213121212
131212121312121213143424}

\textbf{Length 231} \seqsplit{1212121312121213121212131212121312121213121212141212121312121213121212131212121312121213121212141212121
3121212131212121312121213121212131212121412121213121212131212121312121213121212131212121412121213121212
1312121213121212132142324}

\textbf{Length 232} \seqsplit{1212121312121213121212131212121312121213121212141212121312121213121212131212121312121213121212141212121
3121212131212121312121213121212131212121412121213121212131212121312121213121212131212121412121213121212
13121212131212121314134214}

\textbf{Length 233} \seqsplit{1212121312121213121212131212121312121213121212141212121312121213121212131212121312121213121212141212121
3121212131212121312121213121212131212121412121213121212131212121312121213121212131212121412121213121212
131212121312121213121423124}

\textbf{Length 234} \seqsplit{1212121312121213121212131212121312121213121212141212121312121213121212131212121312121213121212141212121
3121212131212121312121213121212131212121412121213121212131212121312121213121212131212121412121213121212
1312121213121212131212431424}

\textbf{Length 235} \seqsplit{1212121312121213121212131212121312121213121212141212121312121213121212131212121312121213121212141212121
3121212131212121312121213121212131212121412121213121212131212121312121213121212131212121412121213121212
13121212131212121312121242324}

\textbf{Length 236} \seqsplit{1212121312121213121212131212121312121213121212141212121312121213121212131212121312121213121212141212121
3121212131212121312121213121212131212121412121213121212131212121312121213121212131212121412121213121212
131212121312121213121241314214}

\textbf{Length 237} \seqsplit{1212121312121213121212131212121312121213121212141212121312121213121212131212121312121213121212141212121
3121212131212121312121213121212131212121412121213121212131212121312121213121212131212121412121213121212
1312121213121212131212124123214}

\textbf{Length 238} \seqsplit{1212121312121213121212131212121312121213121212141212121312121213121212131212121312121213121212141212121
3121212131212121312121213121212131212121412121213121212131212121312121213121212131212121412121213121212
13121212131212121312121241312134}

\textbf{Length 239} \seqsplit{1212121312121213121212131212121312121213121212141212121312121213121212131212121312121213121212141212121
3121212131212121312121213121212131212121412121213121212131212121312121213121212131212121412121213121212
131212121312121213121212421323134}

\textbf{Length 240} \seqsplit{1212121312121213121212131212121312121213121212141212121312121213121212131212121312121213121212141212121
3121212131212121312121213121212131212121412121213121212131212121312121213121212131212121412121213121212
1312121213121212131212121324343434}

\textbf{Length 241} \seqsplit{1212121312121213121212131212121312121213121212141212121312121213121212131212121312121213121212141212121
3121212131212121312121213121212131212121412121213121212131212121312121213121212131212121412121213121212
13121212131212121312121213121242434}

\textbf{Length 242} \seqsplit{1212121312121213121212131212121312121213121212141212121312121213121212131212121312121213121212141212121
3121212131212121312121213121212131212121412121213121212131212121312121213121212131212121412121213121212
131212121312121213121212131243143434}

\textbf{Length 243} \seqsplit{1212121312121213121212131212121312121213121212141212121312121213121212131212121312121213121212141212121
3121212131212121312121213121212131212121412121213121212131212121312121213121212131212121412121213121212
1312121213121212131212121312124124314}

\textbf{Length 244} \seqsplit{1212121312121213121212131212121312121213121212141212121312121213121212131212121312121213121212141212121
3121212131212121312121213121212131212121412121213121212131212121312121213121212131212121412121213121212
13121212131212121312121213124123123424}

\textbf{Length 245} \seqsplit{1212121312121213121212131212121312121213121212141212121312121213121212131212121312121213121212141212121
3121212131212121312121213121212131212121412121213121212131212121312121213121212131212121412121213121212
131212121312121213121212131212413423424}

\textbf{Length 246} \seqsplit{1212121312121213121212131212121312121213121212141212121312121213121212131212121312121213121212141212121
3121212131212121312121213121212131212121412121213121212131212121312121213121212131212121412121213121212
1312121213121212131212121312121214343434}

\textbf{Length 247} \seqsplit{1212121312121213121212131212121312121213121212141212121312121213121212131212121312121213121212141212121
3121212131212121312121213121212131212121412121213121212131212121312121213121212131212121412121213121212
13121212131212121312121213121212142412134}

\textbf{Length 248} \seqsplit{1212121312121213121212131212121312121213121212141212121312121213121212131212121312121213121212141212121
3121212131212121312121213121212131212121412121213121212131212121312121213121212131212121412121213121212
131212121312121213121212131212121412323424}

\textbf{Length 249} \seqsplit{1212121312121213121212131212121312121213121212141212121312121213121212131212121312121213121212141212121
3121212131212121312121213121212131212121412121213121212131212121312121213121212131212121412121213121212
1312121213121212131212121312121214124121314}

\textbf{Length 250} \seqsplit{1212121312121213121212131212121312121213121212141212121312121213121212131212121312121213121212141212121
3121212131212121312121213121212131212121412121213121212131212121312121213121212131212121412121213121212
13121212131212121312121213121212141213231424}

\textbf{Length 251} \seqsplit{1212121312121213121212131212121312121213121212141212121312121213121212131212121312121213121212141212121
3121212131212121312121213121212131212121412121213121212131212121312121213121212131212121412121213121212
131212121312121213121212131212121412121312424}

\textbf{Length 252} \seqsplit{1212121312121213121212131212121312121213121212141212121312121213121212131212121312121213121212141212121
3121212131212121312121213121212131212121412121213121212131212121312121213121212131212121412121213121212
1312121213121212131212121312121214121213431434}

\textbf{Length 253} \seqsplit{1212121312121213121212131212121312121213121212141212121312121213121212131212121312121213121212141212121
3121212131212121312121213121212131212121412121213121212131212121312121213121212131212121412121213121212
13121212131212121312121213121212141212121412134}

\textbf{Length 254} \seqsplit{1212121312121213121212131212121312121213121212141212121312121213121212131212121312121213121212141212121
3121212131212121312121213121212131212121412121213121212131212121312121213121212131212121412121213121212
131212121312121213121212131212121412121213213424}

\textbf{Length 255} \seqsplit{1212121312121213121212131212121312121213121212141212121312121213121212131212121312121213121212141212121
3121212131212121312121213121212131212121412121213121212131212121312121213121212131212121412121213121212
1312121213121212131212121312121214121212131341234}

\textbf{Length 256} \seqsplit{1212121312121213121212131212121312121213121212141212121312121213121212131212121312121213121212141212121
3121212131212121312121213121212131212121412121213121212131212121312121213121212131212121412121213121212
13121212131212121312121213121212141212121343124314}

\textbf{Length 257} \seqsplit{1212121312121213121212131212121312121213121212141212121312121213121212131212121312121213121212141212121
3121212131212121312121213121212131212121412121213121212131212121312121213121212131212121412121213121212
131212121312121213121212131212121412121213131232424}

\textbf{Length 258} \seqsplit{1212121312121213121212131212121312121213121212141212121312121213121212131212121312121213121212141212121
3121212131212121312121213121212131212121412121213121212131212121312121213121212131212121412121213121212
1312121213121212131212121312121214121212131212123424}

\textbf{Length 259} \seqsplit{1212121312121213121212131212121312121213121212141212121312121213121212131212121312121213121212141212121
3121212131212121312121213121212131212121412121213121212131212121312121213121212131212121412121213121212
13121212131212121312121213121212141212121312123412314}

\textbf{Length 260} \seqsplit{1212121312121213121212131212121312121213121212141212121312121213121212131212121312121213121212141212121
3121212131212121312121213121212131212121412121213121212131212121312121213121212131212121412121213121212
131212121312121213121212131212121412121213121212134214}

\textbf{Length 261} \seqsplit{1212121312121213121212131212121312121213121212141212121312121213121212131212121312121213121212141212121
3121212131212121312121213121212131212121412121213121212131212121312121213121212131212121412121213121212
1312121213121212131212121312121214121212131212141214324}

\textbf{Length 262} \seqsplit{1212121312121213121212131212121312121213121212141212121312121213121212131212121312121213121212141212121
3121212131212121312121213121212131212121412121213121212131212121312121213121212131212121412121213121212
13121212131212121312121213121212141212121312123131231424}

\textbf{Length 263} \seqsplit{1212121312121213121212131212121312121213121212141212121312121213121212131212121312121213121212141212121
3121212131212121312121213121212131212121412121213121212131212121312121213121212131212121412121213121212
131212121312121213121212131212121412121213121212312143434}

\textbf{Length 264} \seqsplit{1212121312121213121212131212121312121213121212141212121312121213121212131212121312121213121212141212121
3121212131212121312121213121212131212121412121213121212131212121312121213121212131212121412121213121212
1312121213121212131212121312121214121212131212121313131424}

\textbf{Length 265} \seqsplit{1212121312121213121212131212121312121213121212141212121312121213121212131212121312121213121212141212121
3121212131212121312121213121212131212121412121213121212131212121312121213121212131212121412121213121212
13121212131212121312121213121212141212121312121213121434314}

\textbf{Length 266} \seqsplit{1212121312121213121212131212121312121213121212141212121312121213121212131212121312121213121212141212121
3121212131212121312121213121212131212121412121213121212131212121312121213121212131212121412121213121212
131212121312121213121212131212121412121213121212131231241234}

\textbf{Length 267} \seqsplit{1212121312121213121212131212121312121213121212141212121312121213121212131212121312121213121212141212121
3121212131212121312121213121212131212121412121213121212131212121312121213121212131212121412121213121212
1312121213121212131212121312121214121212131212121312142123434}

\textbf{Length 268} \seqsplit{1212121312121213121212131212121312121213121212141212121312121213121212131212121312121213121212141212121
3121212131212121312121213121212131212121412121213121212131212121312121213121212131212121412121213121212
13121212131212121312121213121212141212121312121213121214213134}

\textbf{Length 269} \seqsplit{1212121312121213121212131212121312121213121212141212121312121213121212131212121312121213121212141212121
3121212131212121312121213121212131212121412121213121212131212121312121213121212131212121412121213121212
131212121312121213121212131212121412121213121212131212124314314}

\textbf{Length 270} \seqsplit{1212121312121213121212131212121312121213121212141212121312121213121212131212121312121213121212141212121
3121212131212121312121213121212131212121412121213121212131212121312121213121212131212121412121213121212
1312121213121212131212121312121214121212131212121312121213132424}

\textbf{Length 271} \seqsplit{1212121312121213121212131212121312121213121212141212121312121213121212131212121312121213121212141212121
3121212131212121312121213121212131212121412121213121212131212121312121213121212131212121412121213121212
13121212131212121312121213121212141212121312121213121212421231434}

\textbf{Length 272} \seqsplit{1212121312121213121212131212121312121213121212141212121312121213121212131212121312121213121212141212121
3121212131212121312121213121212131212121412121213121212131212121312121213121212131212121412121213121212
131212121312121213121212131212121412121213121212131212121312141234}

\textbf{Length 273} \seqsplit{1212121312121213121212131212121312121213121212141212121312121213121212131212121312121213121212141212121
3121212131212121312121213121212131212121412121213121212131212121312121213121212131212121412121213121212
1312121213121212131212121312121214121212131212121312121213131431434}

\textbf{Length 274} \seqsplit{1212121312121213121212131212121312121213121212141212121312121213121212131212121312121213121212141212121
3121212131212121312121213121212131212121412121213121212131212121312121213121212131212121412121213121212
13121212131212121312121213121212141212121312121213121212131212321424}

\textbf{Length 275} \seqsplit{1212121312121213121212131212121312121213121212141212121312121213121212131212121312121213121212141212121
3121212131212121312121213121212131212121412121213121212131212121312121213121212131212121412121213121212
131212121312121213121212131212121412121213121212131212121312123143434}

\textbf{Length 276} \seqsplit{1212121312121213121212131212121312121213121212141212121312121213121212131212121312121213121212141212121
3121212131212121312121213121212131212121412121213121212131212121312121213121212131212121412121213121212
1312121213121212131212121312121214121212131212121312121213121212412314}

\textbf{Length 277} \seqsplit{1212121312121213121212131212121312121213121212141212121312121213121212131212121312121213121212141212121
3121212131212121312121213121212131212121412121213121212131212121312121213121212131212121412121213121212
13121212131212121312121213121212141212121312121213121212131212124321324}

\textbf{Length 278} \seqsplit{1212121312121213121212131212121312121213121212141212121312121213121212131212121312121213121212141212121
3121212131212121312121213121212131212121412121213121212131212121312121213121212131212121412121213121212
131212121312121213121212131212121412121213121212131212121312121241313234}

\textbf{Length 279} \seqsplit{1212121312121213121212131212121312121213121212141212121312121213121212131212121312121213121212141212121
3121212131212121312121213121212131212121412121213121212131212121312121213121212131212121412121213121212
1312121213121212131212121312121214121212131212121312121213121212131231424}

\textbf{Length 280} \seqsplit{1212121312121213121212131212121312121213121212141212121312121213121212131212121312121213121212141212121
3121212131212121312121213121212131212121412121213121212131212121312121213121212131212121412121213121212
13121212131212121312121213121212141212121312121213121212131212121314124324}

\textbf{Length 281} \seqsplit{1212121312121213121212131212121312121213121212141212121312121213121212131212121312121213121212141212121
3121212131212121312121213121212131212121412121213121212131212121312121213121212131212121412121213121212
131212121312121213121212131212121412121213121212131212121312121213143121434}

\textbf{Length 282} \seqsplit{1212121312121213121212131212121312121213121212141212121312121213121212131212121312121213121212141212121
3121212131212121312121213121212131212121412121213121212131212121312121213121212131212121412121213121212
1312121213121212131212121312121214121212131212121312121213121212131232143434}

\textbf{Length 283} \seqsplit{1212121312121213121212131212121312121213121212141212121312121213121212131212121312121213121212141212121
3121212131212121312121213121212131212121412121213121212131212121312121213121212131212121412121213121212
13121212131212121312121213121212141212121312121213121212131212121312124321434}

\textbf{Length 284} \seqsplit{1212121312121213121212131212121312121213121212141212121312121213121212131212121312121213121212141212121
3121212131212121312121213121212131212121412121213121212131212121312121213121212131212121412121213121212
131212121312121213121212131212121412121213121212131212121312121213121241243214}

\textbf{Length 285} \seqsplit{1212121312121213121212131212121312121213121212141212121312121213121212131212121312121213121212141212121
3121212131212121312121213121212131212121412121213121212131212121312121213121212131212121412121213121212
1312121213121212131212121312121214121212131212121312121213121212131212421232434}

\textbf{Length 286} \seqsplit{1212121312121213121212131212121312121213121212141212121312121213121212131212121312121213121212141212121
3121212131212121312121213121212131212121412121213121212131212121312121213121212131212121412121213121212
13121212131212121312121213121212141212121312121213121212131212121312121213243434}

\textbf{Length 287} \seqsplit{1212121312121213121212131212121312121213121212141212121312121213121212131212121312121213121212141212121
3121212131212121312121213121212131212121412121213121212131212121312121213121212131212121412121213121212
131212121312121213121212131212121412121213121212131212121312121213121212131212424}

\textbf{Length 288} \seqsplit{1212121312121213121212131212121312121213121212141212121312121213121212131212121312121213121212141212121
3121212131212121312121213121212131212121412121213121212131212121312121213121212131212121412121213121212
1312121213121212131212121312121214121212131212121312121213121212131212121312431434}

\textbf{Length 289} \seqsplit{1212121312121213121212131212121312121213121212141212121312121213121212131212121312121213121212141212121
3121212131212121312121213121212131212121412121213121212131212121312121213121212131212121412121213121212
13121212131212121312121213121212141212121312121213121212131212121312121213141313124}

\textbf{Length 290} \seqsplit{1212121312121213121212131212121312121213121212141212121312121213121212131212121312121213121212141212121
3121212131212121312121213121212131212121412121213121212131212121312121213121212131212121412121213121212
131212121312121213121212131212121412121213121212131212121312121213121212131242314324}

\textbf{Length 291} \seqsplit{1212121312121213121212131212121312121213121212141212121312121213121212131212121312121213121212141212121
3121212131212121312121213121212131212121412121213121212131212121312121213121212131212121412121213121212
1312121213121212131212121312121214121212131212121312121213121212131212121314123121234}

\textbf{Length 292} \seqsplit{1212121312121213121212131212121312121213121212141212121312121213121212131212121312121213121212141212121
3121212131212121312121213121212131212121412121213121212131212121312121213121212131212121412121213121212
13121212131212121312121213121212141212121312121213121212131212121312121213121212143434}

\textbf{Length 293} \seqsplit{1212121312121213121212131212121312121213121212141212121312121213121212131212121312121213121212141212121
3121212131212121312121213121212131212121412121213121212131212121312121213121212131212121412121213121212
131212121312121213121212131212121412121213121212131212121312121213121212131212412121214}

\textbf{Length 294} \seqsplit{1212121312121213121212131212121312121213121212141212121312121213121212131212121312121213121212141212121
3121212131212121312121213121212131212121412121213121212131212121312121213121212131212121412121213121212
1312121213121212131212121312121214121212131212121312121213121212131212121312121214234324}

\textbf{Length 295} \seqsplit{1212121312121213121212131212121312121213121212141212121312121213121212131212121312121213121212141212121
3121212131212121312121213121212131212121412121213121212131212121312121213121212131212121412121213121212
13121212131212121312121213121212141212121312121213121212131212121312121213121241213212314}

\textbf{Length 296} \seqsplit{1212121312121213121212131212121312121213121212141212121312121213121212131212121312121213121212141212121
3121212131212121312121213121212131212121412121213121212131212121312121213121212131212121412121213121212
131212121312121213121212131212121412121213121212131212121312121213121212131212121412343124}

\textbf{Length 297} \seqsplit{1212121312121213121212131212121312121213121212141212121312121213121212131212121312121213121212141212121
3121212131212121312121213121212131212121412121213121212131212121312121213121212131212121412121213121212
1312121213121212131212121312121214121212131212121312121213121212131212121312124121312121314}

\textbf{Length 298} \seqsplit{1212121312121213121212131212121312121213121212141212121312121213121212131212121312121213121212141212121
3121212131212121312121213121212131212121412121213121212131212121312121213121212131212121412121213121212
13121212131212121312121213121212141212121312121213121212131212121312121213121212141212134314}

\textbf{Length 299} \seqsplit{1212121312121213121212131212121312121213121212141212121312121213121212131212121312121213121212141212121
3121212131212121312121213121212131212121412121213121212131212121312121213121212131212121412121213121212
131212121312121213121212131212121412121213121212131212121312121213121212131212121412131232434}

\textbf{Length 300} \seqsplit{1212121312121213121212131212121312121213121212141212121312121213121212131212121312121213121212141212121
3121212131212121312121213121212131212121412121213121212131212121312121213121212131212121412121213121212
1312121213121212131212121312121214121212131212121312121213121212131212121312121214121212134324}

\textbf{Length 301} \seqsplit{1212121312121213121212131212121312121213121212141212121312121213121212131212121312121213121212141212121
3121212131212121312121213121212131212121412121213121212131212121312121213121212131212121412121213121212
13121212131212121312121213121212141212121312121213121212131212121312121213121212141212312324324}

\textbf{Length 302} \seqsplit{1212121312121213121212131212121312121213121212141212121312121213121212131212121312121213121212141212121
3121212131212121312121213121212131212121412121213121212131212121312121213121212131212121412121213121212
131212121312121213121212131212121412121213121212131212121312121213121212131212121412121234212324}

\textbf{Length 303} \seqsplit{1212121312121213121212131212121312121213121212141212121312121213121212131212121312121213121212141212121
3121212131212121312121213121212131212121412121213121212131212121312121213121212131212121412121213121212
1312121213121212131212121312121214121212131212121312121213121212131212121312121214121212312132434}

\textbf{Length 304} \seqsplit{1212121312121213121212131212121312121213121212141212121312121213121212131212121312121213121212141212121
3121212131212121312121213121212131212121412121213121212131212121312121213121212131212121412121213121212
13121212131212121312121213121212141212121312121213121212131212121312121213121212141212121313231434}

\textbf{Length 305} \seqsplit{1212121312121213121212131212121312121213121212141212121312121213121212131212121312121213121212141212121
3121212131212121312121213121212131212121412121213121212131212121312121213121212131212121412121213121212
131212121312121213121212131212121412121213121212131212121312121213121212131212121412121213121312434}

\textbf{Length 306} \seqsplit{1212121312121213121212131212121312121213121212141212121312121213121212131212121312121213121212141212121
3121212131212121312121213121212131212121412121213121212131212121312121213121212131212121412121213121212
1312121213121212131212121312121214121212131212121312121213121212131212121312121214121212131212143124}

\textbf{Length 307} \seqsplit{1212121312121213121212131212121312121213121212141212121312121213121212131212121312121213121212141212121
3121212131212121312121213121212131212121412121213121212131212121312121213121212131212121412121213121212
13121212131212121312121213121212141212121312121213121212131212121312121213121212141212121312312421234}

\textbf{Length 308} \seqsplit{1212121312121213121212131212121312121213121212141212121312121213121212131212121312121213121212141212121
3121212131212121312121213121212131212121412121213121212131212121312121213121212131212121412121213121212
131212121312121213121212131212121412121213121212131212121312121213121212131212121412121213121232131434}

\textbf{Length 309} \seqsplit{1212121312121213121212131212121312121213121212141212121312121213121212131212121312121213121212141212121
3121212131212121312121213121212131212121412121213121212131212121312121213121212131212121412121213121212
1312121213121212131212121312121214121212131212121312121213121212131212121312121214121212131212123121434}

\textbf{Length 310} \seqsplit{1212121312121213121212131212121312121213121212141212121312121213121212131212121312121213121212141212121
3121212131212121312121213121212131212121412121213121212131212121312121213121212131212121412121213121212
1312121213121212131212121312121214121212131212121312121213121212131212121312121214121212131212141312313
4}

\textbf{Length 311} \seqsplit{1212121312121213121212131212121312121213121212141212121312121213121212131212121312121213121212141212121
3121212131212121312121213121212131212121412121213121212131212121312121213121212131212121412121213121212
1312121213121212131212121312121214121212131212121312121213121212131212121312121214121212131212123212124
34}

\textbf{Length 312} \seqsplit{1212121312121213121212131212121312121213121212141212121312121213121212131212121312121213121212141212121
3121212131212121312121213121212131212121412121213121212131212121312121213121212131212121412121213121212
1312121213121212131212121312121214121212131212121312121213121212131212121312121214121212131212121313243
214}

\textbf{Length 313} \seqsplit{1212121312121213121212131212121312121213121212141212121312121213121212131212121312121213121212141212121
3121212131212121312121213121212131212121412121213121212131212121312121213121212131212121412121213121212
1312121213121212131212121312121214121212131212121312121213121212131212121312121214121212131212121312142
1234}

\textbf{Length 314} \seqsplit{1212121312121213121212131212121312121213121212141212121312121213121212131212121312121213121212141212121
3121212131212121312121213121212131212121412121213121212131212121312121213121212131212121412121213121212
1312121213121212131212121312121214121212131212121312121213121212131212121312121214121212131212121313124
31214}

\textbf{Length 315} \seqsplit{1212121312121213121212131212121312121213121212141212121312121213121212131212121312121213121212141212121
3121212131212121312121213121212131212121412121213121212131212121312121213121212131212121412121213121212
1312121213121212131212121312121214121212131212121312121213121212131212121312121214121212131212121312123
132434}

\textbf{Length 316} \seqsplit{1212121312121213121212131212121312121213121212141212121312121213121212131212121312121213121212141212121
3121212131212121312121213121212131212121412121213121212131212121312121213121212131212121412121213121212
1312121213121212131212121312121214121212131212121312121213121212131212121312121214121212131212121312123
2143214}

\textbf{Length 317} \seqsplit{1212121312121213121212131212121312121213121212141212121312121213121212131212121312121213121212141212121
3121212131212121312121213121212131212121412121213121212131212121312121213121212131212121412121213121212
1312121213121212131212121312121214121212131212121312121213121212131212121312121214121212131212121312121
24212314}

\textbf{Length 318} \seqsplit{1212121312121213121212131212121312121213121212141212121312121213121212131212121312121213121212141212121
3121212131212121312121213121212131212121412121213121212131212121312121213121212131212121412121213121212
1312121213121212131212121312121214121212131212121312121213121212131212121312121214121212131212121312121
241232324}

\textbf{Length 319} \seqsplit{1212121312121213121212131212121312121213121212141212121312121213121212131212121312121213121212141212121
3121212131212121312121213121212131212121412121213121212131212121312121213121212131212121412121213121212
1312121213121212131212121312121214121212131212121312121213121212131212121312121214121212131212121312121
2131314314}

\textbf{Length 320} \seqsplit{1212121312121213121212131212121312121213121212141212121312121213121212131212121312121213121212141212121
3121212131212121312121213121212131212121412121213121212131212121312121213121212131212121412121213121212
1312121213121212131212121312121214121212131212121312121213121212131212121312121214121212131212121312121
21323232434}

\textbf{Length 321} \seqsplit{1212121312121213121212131212121312121213121212141212121312121213121212131212121312121213121212141212121
3121212131212121312121213121212131212121412121213121212131212121312121213121212131212121412121213121212
1312121213121212131212121312121214121212131212121312121213121212131212121312121214121212131212121312121
213121231434}

\textbf{Length 322} \seqsplit{1212121312121213121212131212121312121213121212141212121312121213121212131212121312121213121212141212121
3121212131212121312121213121212131212121412121213121212131212121312121213121212131212121412121213121212
1312121213121212131212121312121214121212131212121312121213121212131212121312121214121212131212121312121
2131232132434}

\textbf{Length 323} \seqsplit{1212121312121213121212131212121312121213121212141212121312121213121212131212121312121213121212141212121
3121212131212121312121213121212131212121412121213121212131212121312121213121212131212121412121213121212
1312121213121212131212121312121214121212131212121312121213121212131212121312121214121212131212121312121
21312314212324}

\textbf{Length 324} \seqsplit{1212121312121213121212131212121312121213121212141212121312121213121212131212121312121213121212141212121
3121212131212121312121213121212131212121412121213121212131212121312121213121212131212121412121213121212
1312121213121212131212121312121214121212131212121312121213121212131212121312121214121212131212121312121
213123121321434}

\textbf{Length 325} \seqsplit{1212121312121213121212131212121312121213121212141212121312121213121212131212121312121213121212141212121
3121212131212121312121213121212131212121412121213121212131212121312121213121212131212121412121213121212
1312121213121212131212121312121214121212131212121312121213121212131212121312121214121212131212121312121
2131212312124314}

\textbf{Length 326} \seqsplit{1212121312121213121212131212121312121213121212141212121312121213121212131212121312121213121212141212121
3121212131212121312121213121212131212121412121213121212131212121312121213121212131212121412121213121212
1312121213121212131212121312121214121212131212121312121213121212131212121312121214121212131212121312121
21312121213232434}

\textbf{Length 327} \seqsplit{1212121312121213121212131212121312121213121212141212121312121213121212131212121312121213121212141212121
3121212131212121312121213121212131212121412121213121212131212121312121213121212131212121412121213121212
1312121213121212131212121312121214121212131212121312121213121212131212121312121214121212131212121312121
213121212131431214}

\textbf{Length 328} \seqsplit{1212121312121213121212131212121312121213121212141212121312121213121212131212121312121213121212141212121
3121212131212121312121213121212131212121412121213121212131212121312121213121212131212121412121213121212
1312121213121212131212121312121214121212131212121312121213121212131212121312121214121212131212121312121
2131212121312321434}

\textbf{Length 329} \seqsplit{1212121312121213121212131212121312121213121212141212121312121213121212131212121312121213121212141212121
3121212131212121312121213121212131212121412121213121212131212121312121213121212131212121412121213121212
1312121213121212131212121312121214121212131212121312121213121212131212121312121214121212131212121312121
21312121213121243214}

\textbf{Length 330} \seqsplit{1212121312121213121212131212121312121213121212141212121312121213121212131212121312121213121212141212121
3121212131212121312121213121212131212121412121213121212131212121312121213121212131212121412121213121212
1312121213121212131212121312121214121212131212121312121213121212131212121312121214121212131212121312121
213121212131231214314}

\textbf{Length 331} \seqsplit{1212121312121213121212131212121312121213121212141212121312121213121212131212121312121213121212141212121
3121212131212121312121213121212131212121412121213121212131212121312121213121212131212121412121213121212
1312121213121212131212121312121214121212131212121312121213121212131212121312121214121212131212121312121
2131212121312124212324}

\textbf{Length 332} \seqsplit{1212121312121213121212131212121312121213121212141212121312121213121212131212121312121213121212141212121
3121212131212121312121213121212131212121412121213121212131212121312121213121212131212121412121213121212
1312121213121212131212121312121214121212131212121312121213121212131212121312121214121212131212121312121
21312121213121212132434}

\textbf{Length 333} \seqsplit{1212121312121213121212131212121312121213121212141212121312121213121212131212121312121213121212141212121
3121212131212121312121213121212131212121412121213121212131212121312121213121212131212121412121213121212
1312121213121212131212121312121214121212131212121312121213121212131212121312121214121212131212121312121
213121212131212412121324}

\textbf{Length 334} \seqsplit{1212121312121213121212131212121312121213121212141212121312121213121212131212121312121213121212141212121
3121212131212121312121213121212131212121412121213121212131212121312121213121212131212121412121213121212
1312121213121212131212121312121214121212131212121312121213121212131212121312121214121212131212121312121
2131212121312121213124314}

\textbf{Length 335} \seqsplit{1212121312121213121212131212121312121213121212141212121312121213121212131212121312121213121212141212121
3121212131212121312121213121212131212121412121213121212131212121312121213121212131212121412121213121212
1312121213121212131212121312121214121212131212121312121213121212131212121312121214121212131212121312121
21312121213121241212323134}

\textbf{Length 336} \seqsplit{1212121312121213121212131212121312121213121212141212121312121213121212131212121312121213121212141212121
3121212131212121312121213121212131212121412121213121212131212121312121213121212131212121412121213121212
1312121213121212131212121312121214121212131212121312121213121212131212121312121214121212131212121312121
213121212131212412131313234}

\textbf{Length 337} \seqsplit{1212121312121213121212131212121312121213121212141212121312121213121212131212121312121213121212141212121
3121212131212121312121213121212131212121412121213121212131212121312121213121212131212121412121213121212
1312121213121212131212121312121214121212131212121312121213121212131212121312121214121212131212121312121
2131212121312124121231313124}

\textbf{Length 338} \seqsplit{1212121312121213121212131212121312121213121212141212121312121213121212131212121312121213121212141212121
3121212131212121312121213121212131212121412121213121212131212121312121213121212131212121412121213121212
1312121213121212131212121312121214121212131212121312121213121212131212121312121214121212131212121312121
21312121213121212131212121434}

\textbf{Length 339} \seqsplit{1212121312121213121212131212121312121213121212141212121312121213121212131212121312121213121212141212121
3121212131212121312121213121212131212121412121213121212131212121312121213121212131212121412121213121212
1312121213121212131212121312121214121212131212121312121213121212131212121312121214121212131212121312121
213121212131212121312412132324}

\textbf{Length 340} \seqsplit{1212121312121213121212131212121312121213121212141212121312121213121212131212121312121213121212141212121
3121212131212121312121213121212131212121412121213121212131212121312121213121212131212121412121213121212
1312121213121212131212121312121214121212131212121312121213121212131212121312121214121212131212121312121
2131212121312121213124123132134}

\textbf{Length 341} \seqsplit{1212121312121213121212131212121312121213121212141212121312121213121212131212121312121213121212141212121
3121212131212121312121213121212131212121412121213121212131212121312121213121212131212121412121213121212
1312121213121212131212121312121214121212131212121312121213121212131212121312121214121212131212121312121
21312121213121212131241213123214}

\textbf{Length 342} \seqsplit{1212121312121213121212131212121312121213121212141212121312121213121212131212121312121213121212141212121
3121212131212121312121213121212131212121412121213121212131212121312121213121212131212121412121213121212
1312121213121212131212121312121214121212131212121312121213121212131212121312121214121212131212121312121
213121212131212121312412121312124}

\textbf{Length 343} \seqsplit{1212121312121213121212131212121312121213121212141212121312121213121212131212121312121213121212141212121
3121212131212121312121213121212131212121412121213121212131212121312121213121212131212121412121213121212
1312121213121212131212121312121214121212131212121312121213121212131212121312121214121212131212121312121
2131212121312121213124121212321324}

\textbf{Length 344} \seqsplit{1212121312121213121212131212121312121213121212141212121312121213121212131212121312121213121212141212121
3121212131212121312121213121212131212121412121213121212131212121312121213121212131212121412121213121212
1312121213121212131212121312121214121212131212121312121213121212131212121312121214121212131212121312121
21312121213121212131212121413131314}

\textbf{Length 345} \seqsplit{1212121312121213121212131212121312121213121212141212121312121213121212131212121312121213121212141212121
3121212131212121312121213121212131212121412121213121212131212121312121213121212131212121412121213121212
1312121213121212131212121312121214121212131212121312121213121212131212121312121214121212131212121312121
213121212131212121312121214121312324}

\textbf{Length 346} \seqsplit{1212121312121213121212131212121312121213121212141212121312121213121212131212121312121213121212141212121
3121212131212121312121213121212131212121412121213121212131212121312121213121212131212121412121213121212
1312121213121212131212121312121214121212131212121312121213121212131212121312121214121212131212121312121
2131212121312121213121212141212132124}

\textbf{Length 347} \seqsplit{1212121312121213121212131212121312121213121212141212121312121213121212131212121312121213121212141212121
3121212131212121312121213121212131212121412121213121212131212121312121213121212131212121412121213121212
1312121213121212131212121312121214121212131212121312121213121212131212121312121214121212131212121312121
21312121213121212131212121412313132324}

\textbf{Length 348} \seqsplit{1212121312121213121212131212121312121213121212141212121312121213121212131212121312121213121212141212121
3121212131212121312121213121212131212121412121213121212131212121312121213121212131212121412121213121212
1312121213121212131212121312121214121212131212121312121213121212131212121312121214121212131212121312121
213121212131212121312121214121213121214}

\textbf{Length 349} \seqsplit{1212121312121213121212131212121312121213121212141212121312121213121212131212121312121213121212141212121
3121212131212121312121213121212131212121412121213121212131212121312121213121212131212121412121213121212
1312121213121212131212121312121214121212131212121312121213121212131212121312121214121212131212121312121
2131212121312121213121212141212123121324}

\textbf{Length 350} \seqsplit{1212121312121213121212131212121312121213121212141212121312121213121212131212121312121213121212141212121
3121212131212121312121213121212131212121412121213121212131212121312121213121212131212121412121213121212
1312121213121212131212121312121214121212131212121312121213121212131212121312121214121212131212121312121
21312121213121212131212121412121213132314}

\textbf{Length 351} \seqsplit{1212121312121213121212131212121312121213121212141212121312121213121212131212121312121213121212141212121
3121212131212121312121213121212131212121412121213121212131212121312121213121212131212121412121213121212
1312121213121212131212121312121214121212131212121312121213121212131212121312121214121212131212121312121
213121212131212121312121214121212131213124}

\textbf{Length 352} \seqsplit{1212121312121213121212131212121312121213121212141212121312121213121212131212121312121213121212141212121
3121212131212121312121213121212131212121412121213121212131212121312121213121212131212121412121213121212
1312121213121212131212121312121214121212131212121312121213121212131212121312121214121212131212121312121
2131212121312121213121212141212123132312124}

\textbf{Length 353} \seqsplit{1212121312121213121212131212121312121213121212141212121312121213121212131212121312121213121212141212121
3121212131212121312121213121212131212121412121213121212131212121312121213121212131212121412121213121212
1312121213121212131212121312121214121212131212121312121213121212131212121312121214121212131212121312121
21312121213121212131212121412121213123131324}

\textbf{Length 354} \seqsplit{1212121312121213121212131212121312121213121212141212121312121213121212131212121312121213121212141212121
3121212131212121312121213121212131212121412121213121212131212121312121213121212131212121412121213121212
1312121213121212131212121312121214121212131212121312121213121212131212121312121214121212131212121312121
213121212131212121312121214121212131212321314}

\textbf{Length 355} \seqsplit{1212121312121213121212131212121312121213121212141212121312121213121212131212121312121213121212141212121
3121212131212121312121213121212131212121412121213121212131212121312121213121212131212121412121213121212
1312121213121212131212121312121214121212131212121312121213121212131212121312121214121212131212121312121
2131212121312121213121212141212121312121231214}

\textbf{Length 356} \seqsplit{1212121312121213121212131212121312121213121212141212121312121213121212131212121312121213121212141212121
3121212131212121312121213121212131212121412121213121212131212121312121213121212131212121412121213121212
1312121213121212131212121312121214121212131212121312121213121212131212121312121214121212131212121312121
21312121213121212131212121412121213123132312314}

\textbf{Length 357} \seqsplit{1212121312121213121212131212121312121213121212141212121312121213121212131212121312121213121212141212121
3121212131212121312121213121212131212121412121213121212131212121312121213121212131212121412121213121212
1312121213121212131212121312121214121212131212121312121213121212131212121312121214121212131212121312121
213121212131212121312121214121212131212123212124}

\textbf{Length 358} \seqsplit{1212121312121213121212131212121312121213121212141212121312121213121212131212121312121213121212141212121
3121212131212121312121213121212131212121412121213121212131212121312121213121212131212121412121213121212
1312121213121212131212121312121214121212131212121312121213121212131212121312121214121212131212121312121
2131212121312121213121212141212121312123121212314}

\textbf{Length 359} \seqsplit{1212121312121213121212131212121312121213121212141212121312121213121212131212121312121213121212141212121
3121212131212121312121213121212131212121412121213121212131212121312121213121212131212121412121213121212
1312121213121212131212121312121214121212131212121312121213121212131212121312121214121212131212121312121
21312121213121212131212121412121213121212131231314}

\textbf{Length 360} \seqsplit{1212121312121213121212131212121312121213121212141212121312121213121212131212121312121213121212141212121
3121212131212121312121213121212131212121412121213121212131212121312121213121212131212121412121213121212
1312121213121212131212121312121214121212131212121312121213121212131212121312121214121212131212121312121
213121212131212121312121214121212131212121323232324}

\textbf{Length 361} \seqsplit{1212121312121213121212131212121312121213121212141212121312121213121212131212121312121213121212141212121
3121212131212121312121213121212131212121412121213121212131212121312121213121212131212121412121213121212
1312121213121212131212121312121214121212131212121312121213121212131212121312121214121212131212121312121
2131212121312121213121212141212121312121213121231324}

\textbf{Length 362} \seqsplit{1212121312121213121212131212121312121213121212141212121312121213121212131212121312121213121212141212121
3121212131212121312121213121212131212121412121213121212131212121312121213121212131212121412121213121212
1312121213121212131212121312121214121212131212121312121213121212131212121312121214121212131212121312121
21312121213121212131212121412121213121212131232132324}

\textbf{Length 363} \seqsplit{1212121312121213121212131212121312121213121212141212121312121213121212131212121312121213121212141212121
3121212131212121312121213121212131212121412121213121212131212121312121213121212131212121412121213121212
1312121213121212131212121312121214121212131212121312121213121212131212121312121214121212131212121312121
213121212131212121312121214121212131212121313131212124}

\textbf{Length 364} \seqsplit{1212121312121213121212131212121312121213121212141212121312121213121212131212121312121213121212141212121
3121212131212121312121213121212131212121412121213121212131212121312121213121212131212121412121213121212
1312121213121212131212121312121214121212131212121312121213121212131212121312121214121212131212121312121
2131212121312121213121212141212121312121213121232312314}

\textbf{Length 365} \seqsplit{1212121312121213121212131212121312121213121212141212121312121213121212131212121312121213121212141212121
3121212131212121312121213121212131212121412121213121212131212121312121213121212131212121412121213121212
1312121213121212131212121312121214121212131212121312121213121212131212121312121214121212131212121312121
21312121213121212131212121412121213121212131212312123214}

\textbf{Length 366} \seqsplit{1212121312121213121212131212121312121213121212141212121312121213121212131212121312121213121212141212121
3121212131212121312121213121212131212121412121213121212131212121312121213121212131212121412121213121212
1312121213121212131212121312121214121212131212121312121213121212131212121312121214121212131212121312121
213121212131212121312121214121212131212121312121213232324}

\textbf{Length 367} \seqsplit{1212121312121213121212131212121312121213121212141212121312121213121212131212121312121213121212141212121
3121212131212121312121213121212131212121412121213121212131212121312121213121212131212121412121213121212
1312121213121212131212121312121214121212131212121312121213121212131212121312121214121212131212121312121
2131212121312121213121212141212121312121213121212131212314}

\textbf{Length 368}\seqsplit{1212121312121213121212131212121312121213121212141212121312121213121212131212121312121213121212141212121
3121212131212121312121213121212131212121412121213121212131212121312121213121212131212121412121213121212
1312121213121212131212121312121214121212131212121312121213121212131212121312121214121212131212121312121
21312121213121212131212121412121213121212131212121312321324}

\textbf{Length 369}\seqsplit{1212121312121213121212131212121312121213121212141212121312121213121212131212121312121213121212141212121
3121212131212121312121213121212131212121412121213121212131212121312121213121212131212121412121213121212
1312121213121212131212121312121214121212131212121312121213121212131212121312121214121212131212121312121
213121212131212121312121214121212131212121312123121231212124}

\textbf{Length 370}\seqsplit{1212121312121213121212131212121312121213121212141212121312121213121212131212121312121213121212141212121
3121212131212121312121213121212131212121412121213121212131212121312121213121212131212121412121213121212
1312121213121212131212121312121214121212131212121312121213121212131212121312121214121212131212121312121
2131212121312121213121212141212121312121213121212131231213214}

\textbf{Length 371}\seqsplit{1212121312121213121212131212121312121213121212141212121312121213121212131212121312121213121212141212121
3121212131212121312121213121212131212121412121213121212131212121312121213121212131212121412121213121212
1312121213121212131212121312121214121212131212121312121213121212131212121312121214121212131212121312121
21312121213121212131212121412121213121212131212121313213212324}

\textbf{Length 372}\seqsplit{1212121312121213121212131212121312121213121212141212121312121213121212131212121312121213121212141212121
3121212131212121312121213121212131212121412121213121212131212121312121213121212131212121412121213121212
1312121213121212131212121312121214121212131212121312121213121212131212121312121214121212131212121312121
213121212131212121312121214121212131212121312121213121212132324}

\textbf{Length 373}
\seqsplit{1212121312121213121212131212121312121213121212141212121312121213121212131212121312121213121212141212121
3121212131212121312121213121212131212121412121213121212131212121312121213121212131212121412121213121212
1312121213121212131212121312121214121212131212121312121213121212131212121312121214121212131212121312121
2131212121312121213121212141212121312121213121212131212321212124}

\textbf{Length 374}
\seqsplit{1212121312121213121212131212121312121213121212141212121312121213121212131212121312121213121212141212121
3121212131212121312121213121212131212121412121213121212131212121312121213121212131212121412121213121212
1312121213121212131212121312121214121212131212121312121213121212131212121312121214121212131212121312121
21312121213121212131212121412121213121212131212121312121213123214}

\textbf{Length 375}
\seqsplit{1212121312121213121212131212121312121213121212141212121312121213121212131212121312121213121212141212121
3121212131212121312121213121212131212121412121213121212131212121312121213121212131212121412121213121212
1312121213121212131212121312121214121212131212121312121213121212131212121312121214121212131212121312121
213121212131212121312121214121212131212121312121213121232312121324}

\textbf{Length 376}
\seqsplit{1212121312121213121212131212121312121213121212141212121312121213121212131212121312121213121212141212121
3121212131212121312121213121212131212121412121213121212131212121312121213121212131212121412121213121212
1312121213121212131212121312121214121212131212121312121213121212131212121312121214121212131212121312121
2131212121312121213121212141212121312121213121212131232121231212124}

\textbf{Length 377}
\seqsplit{1212121312121213121212131212121312121213121212141212121312121213121212131212121312121213121212141212121
3121212131212121312121213121212131212121412121213121212131212121312121213121212131212121412121213121212
1312121213121212131212121312121214121212131212121312121213121212131212121312121214121212131212121312121
21312121213121212131212121412121213121212131212121312312132312132324}

\textbf{Length 378}
\seqsplit{1212121312121213121212131212121312121213121212141212121312121213121212131212121312121213121212141212121
3121212131212121312121213121212131212121412121213121212131212121312121213121212131212121412121213121212
1312121213121212131212121312121214121212131212121312121213121212131212121312121214121212131212121312121
213121212131212121312121214121212131212121312121213121212131212121324}

\textbf{Length 379}
\seqsplit{1212121312121213121212131212121312121213121212141212121312121213121212131212121312121213121212141212121
3121212131212121312121213121212131212121412121213121212131212121312121213121212131212121412121213121212
1312121213121212131212121312121214121212131212121312121213121212131212121312121214121212131212121312121
2131212121312121213121212141212121312121213121212131232123121321212124}

\textbf{Length 380}
\seqsplit{1212121312121213121212131212121312121213121212141212121312121213121212131212121312121213121212141212121
3121212131212121312121213121212131212121412121213121212131212121312121213121212131212121412121213121212
1312121213121212131212121312121214121212131212121312121213121212131212121312121214121212131212121312121
21312121213121212131212121412121213121212131212121313121312121213121314}

\textbf{Length 381}
\seqsplit{1212121312121213121212131212121312121213121212141212121312121213121212131212121312121213121212141212121
3121212131212121312121213121212131212121412121213121212131212121312121213121212131212121412121213121212
1312121213121212131212121312121214121212131212121312121213121212131212121312121214121212131212121312121
213121212131212121312121214121212131212121312121213123121213121321212124}

\textbf{Length 382}
\seqsplit{1212121312121213121212131212121312121213121212141212121312121213121212131212121312121213121212141212121
3121212131212121312121213121212131212121412121213121212131212121312121213121212131212121412121213121212
1312121213121212131212121312121214121212131212121312121213121212131212121312121214121212131212121312121
2131212121312121213121212141212121312121213121212131212121321212131212124}

\textbf{Length 383}
\seqsplit{1212121312121213121212131212121312121213121212141212121312121213121212131212121312121213121212141212121
3121212131212121312121213121212131212121412121213121212131212121312121213121212131212121412121213121212
1312121213121212131212121312121214121212131212121312121213121212131212121312121214121212131212121312121
21312121213121212131212121412121213121232121231213121212132312121213131214}

\textbf{Length 384}
\seqsplit{1212121312121213121212131212121312121213121212141212121312121213121212131212121312121213121212141212121
3121212131212121312121213121212131212121412121213121212131212121312121213121212131212121412121213121212
1312121213121212131212121312121214121212131212121312121213121212131212121312121214121212131212121312121
213121212131212121312121214121212131212121312121213121212131212121312121214}

\end{footnotesize}

\end{document}